\documentclass[final]{IEEEtran}
\newlength \figwidth
\usepackage{cite}
\ifCLASSINFOpdf
  \usepackage[pdftex]{graphicx}
	\usepackage{epstopdf}
  \graphicspath{{../Figures/}}
  \DeclareGraphicsExtensions{.eps}
\else
  \usepackage[dvips,draft]{graphicx}
\fi
\usepackage[cmex10]{amsmath}
\usepackage{amssymb}
\usepackage{pifont}
\usepackage{amsthm}
\usepackage{url}
\usepackage{dsfont}
\usepackage{tabulary}
\usepackage{multirow}
\usepackage{color}

\usepackage{color}
\usepackage{times}
\usepackage{verbatim}
\usepackage{floatflt}
\usepackage{enumerate}
\usepackage{array}
\usepackage{multicol,afterpage,wrapfig} 
\usepackage{tikz}
\usepackage[noend]{algpseudocode}
\makeatletter
\def\BState{\State\hskip-\ALG@thistlm}
\makeatother
\usepackage{amsmath,amssymb,eucal}
\usepackage[utf8]{inputenc}
\usepackage{epsfig}
\usepackage{exscale}
\usepackage{latexsym}
\usepackage{verbatim}
\usepackage{amsfonts}
\usepackage{multirow}
\usepackage{cite}
\usepackage{balance}
\usepackage{color}
\usepackage{transparent}
\usepackage{import}
\usepackage{mathtools}
\usepackage{tikz}
\usetikzlibrary{automata,arrows,positioning,calc}
\usepackage{bbm}
\usepackage[printonlyused,withpage]{acronym}
\usepackage{tabu}
\usepackage{longtable}
\usepackage{balance}
\usepackage{footnote}
\usepackage{tablefootnote}
\usepackage{enumitem}

\usepackage{adjustbox}
\usepackage{multirow}
\usepackage{pifont}
\def\BibTeX{{\rm B\kern-.05em{\sc i\kern-.025em b}\kern-.08em
    T\kern-.1667em\lower.7ex\hbox{E}\kern-.125emX}}
    
\renewcommand{\IEEEQED}{\IEEEQEDopen}

\newcommand*\xbar[1]{%
  \hbox{%
    \vbox{%
      \hrule height 0.5pt 
      \kern0.36ex
      \hbox{%
        \kern-0.12em
        \ensuremath{#1}%
        \kern-0.12em
      }%
    }%
  }%
}

\newfont{\bbb}{msbm10 scaled 500}

\newfont{\bb}{msbm10 scaled 1100}

\newcommand{\Rc}{{\cal R}}

\newcommand{\bx}{{\text{b}}}

\newcommand{\hx}{{\text{h}}}

\newcommand{\nx}{{\text{n}}}
\newcommand{\ox}{{\text{o}}}
\newcommand{\px}{{\text{p}}}

\newcommand{\sx}{{\text{s}}}

\renewcommand{\d}{\mathrm{d}}

\def\e{\text{e}}

\newcommand{\executeiffilenewer}[3]{%
\ifnum\pdfstrcmp{\pdffilemoddate{#1}}%
{\pdffilemoddate{#2}}>0%
{\immediate\write18{#3}}\fi%
}

\usetikzlibrary{arrows,calc}
\allowdisplaybreaks

\IEEEoverridecommandlockouts
\begin{document}
\pagenumbering{gobble}

\def \pcov{\mathcal{P}_\mathrm{cov}}
\def \pcovl{\mathcal{P}_\mathrm{cov}^\mathrm{L}}
\def \pcovn{\mathcal{P}_\mathrm{cov}^\mathrm{N}}
\def \pcovu{\mathcal{P}_\mathrm{cov}^\upsilon}
\def \pcovuav{\mathcal{P}_\mathrm{u}}
\def \pcovc{\mathcal{P}_\mathrm{g}}
\def \sinr{\mathsf{SINR}}
\def \du{d_\mathrm{u}}
\def \dc{d_\mathrm{g}}
\def \db{d_\mathrm{b}}
\def \ru{r_\mathrm{u}}
\def \Ru{R_\mathrm{u}}
\def \rc{r_\mathrm{g}}
\def \rb{r_\mathrm{b}}
\def \al{\alpha_\mathrm{L}}
\def \an{\alpha_\mathrm{N}}
\def \au{\alpha_\upsilon}
\def \bl{\beta_\mathrm{L}}
\def \bn{\beta_\mathrm{N}}
\def \bu{\beta_\upsilon}
\def \bx{\beta_\xi}
\def \pu{\mathrm{P_u}}
\def \pum{\mathrm{P_u^M}}
\def \pcm{\mathrm{P_g^M}}
\def \pc{\mathrm{P_g}}
\def \pcl{\mathrm{P_g^L}}
\def \pcn{\mathrm{P_g^N}}
\def \hu{\mathrm{h_u}}
\def \htx{h_\mathrm{t}}
\def \hrx{h_\mathrm{r}}
\def \hb{\mathrm{h_b}}
\def \hc{\mathrm{h_g}}
\def \pl{\mathcal{P}_\mathrm{L}}
\def \pn{\mathcal{P}_\mathrm{N}}
\def \pup{\mathcal{P}_\upsilon}
\def \px{\mathcal{P}_\xi}
\def \pluu{\mathcal{P}_\mathrm{L}^\mathrm{uu}}
\def \ml{m_\mathrm{L}}
\def \mn{m_\mathrm{N}}
\def \m{m_\upsilon}
\def \muu{\mathrm{m_{uu}}}
\def \muul{\mathrm{m_{uu}^L}}
\def \muun{\mathrm{m_{uu}^N}}
\def \mcb{\mathrm{m_{gb}}}
\def \mcbl{\mathrm{m_{gb}^L}}
\def \mcbn{\mathrm{m_{gb}^N}}
\def \mub{\mathrm{m_{ub}}}
\def \mubl{\mathrm{m_{ub}^L}}
\def \mubn{\mathrm{m_{ub}^N}}
\def \mcul{\mathrm{m_{gu}^L}}
\def \mcun{\mathrm{m_{gu}^N}}
\def \nl{n_\mathrm{L}}
\def \nn{n_\mathrm{N}}
\def \nx{n_\xi}
\def \n{n_\upsilon}
\def \sl{\psi_\mathrm{L}}
\def \sn{\psi_\mathrm{N}}
\def \su{\psi_\upsilon}
\def \sx{\psi_\xi}
\def \ol{\omega_\mathrm{L}}
\def \on{\omega_\mathrm{N}}
\def \ou{\omega_\upsilon}
\def \ox{\omega_\xi}
\def \il{I_\mathrm{L}}
\def \iN{I_\mathrm{N}}
\def \yl{y_\mathrm{L}}
\def \yn{y_\mathrm{N}}
\def \yu{y_\upsilon}
\def \lapi{\mathcal{L}_I}
\def \lapil{\mathcal{L}_{I_\mathrm{L}}}
\def \lapin{\mathcal{L}_{I_\mathrm{N}}}
\def \lapix{\mathcal{L}_{I_\xi}}
\def \lapic{\mathcal{L}_{I_\mathrm{g}}}
\def \lapiu{\mathcal{L}_{I_\mathrm{u}}}
\def \Ul{\Upsilon_\mathrm{L}}
\def \Un{\Upsilon_\mathrm{N}}
\def \Ux{\Upsilon_\xi}
\def \Uxu{\Upsilon_\xi^{\,\upsilon}}
\def \Unu{\Upsilon_\mathrm{N}^\upsilon}
\def \frl{f_{R_\mathrm{b}}^\mathrm{L}}
\def \frn{f_{R_\mathrm{b}}^\mathrm{N}}
\def \frx{f_{R_\mathrm{b}}^{\,\xi}}
\def \laml{\lambda_\mathrm{L}}
\def \lamn{\lambda_\mathrm{N}}

\def \t{\mathrm{T}}
\def \i{I}
\def \ic{I_\mathrm{g}}
\def \icl{I_\mathrm{g}^\mathrm{L}}
\def \icn{I_\mathrm{g}^\mathrm{N}}
\def \iu{I_\mathrm{u}}
\def \iul{I_\mathrm{u}^\mathrm{L}}
\def \iun{I_\mathrm{u}^\mathrm{N}}
\def \ful{\Phi_\mathrm{u}^\mathrm{L}}
\def \fun{\Phi_\mathrm{u}^\mathrm{N}}
\def \fu{\Phi_\mathrm{u}}
\def \fcl{\Phi_\mathrm{g}^\mathrm{L}}
\def \fcN{\Phi_\mathrm{g}^\mathrm{N}}
\def \fc{\Phi_\mathrm{g}}

\def \gb{\mathrm{G_b}}
\def \gu{\mathrm{G_{u}}}
\def \gb{\mathrm{G_b}}
\def \tt{\theta_\mathrm{t}}
\def \tb{\theta_\mathrm{b}}

\def \lamci{\hat{\lambda}_\mathrm{g}}
\def \lamb{\lambda_\mathrm{b}}
\def \lamu{\lambda_\mathrm{u}}
\def \lamuhat{\hat{\lambda}_\mathrm{u}}
\def \lamg{\lambda_\mathrm{g}}
\def \lamul{\lambda_\mathrm{u}^\mathrm{L}}
\def \lamun{\lambda_\mathrm{u}^\mathrm{N}}
\def \lamcl{\lambda_\mathrm{g}^\mathrm{L}}

\def \iccl{I_\mathrm{gg}^\mathrm{L}}
\def \iccn{I_\mathrm{gg}^\mathrm{N}}
\def \icc{I_\mathrm{gg}}
\def \icu{I_\mathrm{gu}}
\def \iuc{I_\mathrm{ug}}
\def \iuu{I_\mathrm{uu}}
\def \iuul{I_\mathrm{uu}^\mathrm{L}}
\def \iuun{I_\mathrm{uu}^\mathrm{N}}
\def \icul{I_\mathrm{gu}^\mathrm{L}}
\def \icun{I_\mathrm{gu}^\mathrm{N}}
\def \iucl{I_\mathrm{ug}^\mathrm{L}}
\def \iucn{I_\mathrm{ug}^\mathrm{N}}
\def \lapiccl{\mathcal{L}_{\iccl}}
\def \lapiccn{\mathcal{L}_{\iccn}}
\def \lapicul{\mathcal{L}_{\icul}}
\def \lapicun{\mathcal{L}_{\icun}}
\def \lapicu{\mathcal{L}_{\icu}}
\def \lapiuc{\mathcal{L}_{\iuc}}
\def \lapicc{\mathcal{L}_{\icc}}
\def \lapiucl{\mathcal{L}_{\iucl}}
\def \lapiucn{\mathcal{L}_{\iucn}}
\def \lapiuu{\mathcal{L}_{\iuu}}
\def \lapiuul{\mathcal{L}_{\iuul}}
\def \lapiuun{\mathcal{L}_{\iuun}}
\def \pcb{\mathsf{p}_\mathrm{gb}}
\def \pub{\mathsf{p}_\mathrm{ub}}
\def \pcu{\mathsf{p}_\mathrm{gu}}
\def \pruu{\mathsf{p}_\mathrm{uu}}
\def \prcb{\mathsf{p}_\mathrm{gb}}
\def \fici{\hat{\Phi}_c}
\def \tl{\tau_\mathrm{L}}
\def \tn{\tau_\mathrm{N}}
\def \acbl{\alpha_{\mathrm{gb}}^\mathrm{L}}
\def \acbn{\alpha_{\mathrm{gb}}^\mathrm{N}}
\def \acul{\alpha_{\mathrm{gu}}^\mathrm{L}}
\def \acun{\alpha_{\mathrm{gu}}^\mathrm{N}}
\def \acb{\alpha_{\mathrm{gb}}}
\def \alcb{\alpha_{\mathrm{gb}}^\mathrm{L}}
\def \ancb{\alpha_{\mathrm{gb}}^\mathrm{N}}
\def \aubl{\alpha_{\mathrm{ub}}^\mathrm{L}}
\def \aubn{\alpha_{\mathrm{ub}}^\mathrm{N}}
\def \auu{\alpha_{\mathrm{uu}}}
\def \auul{\alpha_{\mathrm{uu}}^\mathrm{L}}
\def \auun{\alpha_{\mathrm{uu}}^\mathrm{N}}
\def \zl{\zeta_\mathrm{L}}
\def \zn{\zeta_\mathrm{N}}
\def \ec{\epsilon_\mathrm{g}}
\def \eu{\epsilon_\mathrm{u}}
\def \subl{\psi_\mathrm{ub}^\mathrm{L}}
\def \scbl{\psi_\mathrm{gb}^\mathrm{L}}
\def \scbn{\psi_\mathrm{gb}^\mathrm{N}}
\def \scul{\psi_\mathrm{gu}^\mathrm{L}}
\def \suu{\psi_\mathrm{uu}}
\def \suul{\psi_\mathrm{uu}^\mathrm{L}}
\def \suun{\psi_\mathrm{uu}^\mathrm{N}}
\def \zubl{\zeta_\mathrm{ub}^\mathrm{L}}
\def \zubn{\zeta_\mathrm{ub}^\mathrm{N}}
\def \zcul{\zeta_\mathrm{cu}^\mathrm{L}}
\def \zuu{\zeta_\mathrm{uu}}
\def \zuul{\zeta_\mathrm{uu}^\mathrm{L}}
\def \zuun{\zeta_\mathrm{uu}^\mathrm{N}}
\def \zcb{\zeta_\mathrm{gb}}
\def \zcbl{\zeta_\mathrm{gb}^\mathrm{L}}
\def \zcbn{\zeta_\mathrm{gb}^\mathrm{N}}
\def \pur{\rho_\mathrm{u}}
\def \pcr{\rho_\mathrm{g}}
\def \guu{\mathrm{g_{uu}}}
\def \gcu{\mathrm{g_{gu}}}
\def \gub{\mathrm{g_{ub}}}
\def \gubi{\mathrm{g_{ub}^{(i)}}}
\def \gcb{\mathrm{g_{gb}}}
\def \scb{\psi_\mathrm{gb}}
\def \hub{\mathrm{h_{ub}}}
\def \hcb{\mathrm{h_{gb}}}
\def \hcu{\mathrm{h_{gu}}}
\def \dub{d_\mathrm{{ub}}}
\def \dcb{d_\mathrm{{gb}}}
\def \dcu{d_\mathrm{{gu}}}
\def \publ{\Psi_\mathrm{ub}^\mathrm{L}}
\def \pcul{\Psi_\mathrm{gu}^\mathrm{L}}
\def \pcb{\Psi_\mathrm{gb}}
\def \pcbl{\Psi_\mathrm{gb}^\mathrm{L}}
\def \pcbn{\Psi_\mathrm{gb}^\mathrm{N}}
\def \puu{\Psi_\mathrm{uu}}
\def \puul{\Psi_\mathrm{uu}^\mathrm{L}}
\def \puun{\Psi_\mathrm{uu}^\mathrm{N}}
\def \buu{\beta_\mathrm{uu}}
\def \bcb{\beta_\mathrm{gb}}
\def \bubl{\beta_\mathrm{ub}^\mathrm{L}}
\def \Bx{\mathrm{B_x}}
\def \Cx{\mathcal{C}_\mathrm{x}}
\def \Rx{\mathcal{R}_\mathrm{x}}
 
\def \hb{\mathrm{h}_{\mathrm{b}}}
\def \hu{\mathrm{h}_{\mathrm{u}}}
\def \hg{\mathrm{h}_{\mathrm{g}}}
\def \hx{\mathrm{h}_{\mathrm{x}}}
\def \hy{\mathrm{h}_{\mathrm{y}}}
\def \hxy{\mathrm{h}_{\mathrm{xy}}}
\def \x{\mathrm{x}}
\def \y{\mathrm{y}}
\def \zetaxy{\zeta_{\mathrm{xy}}}
\def \tauxy{\tau_{\mathrm{xy}}}
\def \tauoxy{\tau_{0,\mathrm{xy}}}
\def \gxy{g_{\mathrm{xy}}}
\def \psixy{\psi_{\mathrm{xy}}}
\def \psixyL{\psixy^{\mathrm{L}}}
\def \psixyN{\psixy^{\mathrm{N}}}
\def \L{\mathrm{L}}
\def \N{\mathrm{N}}
\def \a{\mathrm{a}}
\def \x{\mathrm{x}}
\def \y{\mathrm{y}}
\def \dxy{d_\mathrm{xy}}
\def \rxy{r_\mathrm{xy}}
\def \alphaxy{\alpha_\mathrm{xy}}
\def \alphaxyL{\alphaxy^\mathrm{L}}
\def \alphaxyN{\alphaxy^\mathrm{N}}
\def \Px{P_{\mathrm{x}}}
\def \Pxmax{P_{\mathrm{x}}^{\textrm{max}}}
\def \rhox{\rho_{\mathrm{x}}}
\def \epsx{\epsilon_{\mathrm{x}}}
\def \mxy{\mathrm{m}_{\mathrm{xy}}}
\def \mxyL{\mxy^{\mathrm{L}}}
\def \mxyN{\mxy^{\mathrm{N}}}
\def \Ixy{I_{\mathrm{xy}}}
\def \Phib{\Phi_{\mathrm{b}}}
\def \Phig{\Phi_{\mathrm{g}}}
\def \Phiu{\Phi_{\mathrm{u}}}
\def \Phihatg{\hat{\Phi}_{\mathrm{g}}}
\def \Phihatgb{\hat{\Phi}_{\mathrm{gb}}}
\def \Phihatgu{\hat{\Phi}_{\mathrm{gu}}}


\def \pcovuav{\mathcal{C}_\mathrm{u}}
\def \pcovc{\mathcal{C}_\mathrm{g}}
\def \rMuu{\mathrm{r_M}}
\def \ru{\mathrm{r_u}}
\def \Ru{R_\mathrm{u}}
\def \u{\mathrm{u}}
\def \c{\mathrm{g}}
\def \q{\mathrm{q}}
\def \r{\mathrm{r}}
\def \m{\mathrm{m}}
\def \prxy{\mathsf{p}_\mathrm{xy}}
\def \prxyxi{\mathsf{p}_\mathrm{xy}^\xi}
\def \pl{\mathcal{P}^\mathrm{L}}
\def \pn{\mathcal{P}^\mathrm{N}}
\def \rc{\mathrm{r_g}}
\def \d{\mathrm{d}}
\def \L{\mathrm{L}}
\def \N{\mathrm{N}}
\def \pu{P_\mathrm{u}}
\def \pumax{\mathrm{P_u^{max}}}
\def \lapixyxi{\mathcal{L}_{I_\mathrm{xy}^\xi}}
\def \px{P_\mathrm{x}}
\def \sy{\mathrm{s_y}}
\def \sixy{\psi_\mathrm{xy}}
\def \zetxy{\zeta_\mathrm{xy}}
\def \prxi{\mathcal{P}^\xi}
\def \s{\mathrm{s}}
\def \axy{\alpha_\mathrm{xy}}
\def \bxy{\beta_\mathrm{xy}}
\def \hxy{\mathrm{h_{xy}}}
\def \Rc{R_\mathrm{g}}
\def \pc{P_\mathrm{g}}
\def \lamc{\lambda_\c}
\def \D{\mathrm{D}}
\def \sigmau{\sigma_\mathrm{u}}
\def \g{\mathrm{g}}


\newtheorem{Theorem}{\bf Theorem}
\newtheorem{Corollary}{\bf Corollary}
\newtheorem{Remark}{\bf Remark}
\newtheorem{Lemma}{\bf Lemma}
\newtheorem{Proposition}{\bf Proposition}
\newtheorem{Assumption}{\bf Assumption}
\newtheorem{Approximation}{\bf Approximation}
\newtheorem{Definition}{\bf Definition}

\title{UAV-to-UAV Communications in Cellular Networks}
\author{{M.~Mahdi~Azari, Giovanni~Geraci, Adrian~Garcia-Rodriguez, and Sofie~Pollin}
\thanks{M.~M.~Azari was with KU Leuven, Belgium. He is now with Centre Tecnològic de Telecomunicacions de Catalunya (CTTC), Barcelona, Spain. G.~Geraci is with Universitat Pompeu Fabra, Barcelona, Spain. A.~Garcia-Rodriguez is with Nokia Bell Labs, Dublin, Ireland. Sofie~Pollin is with KU Leuven, Belgium.}
\thanks{Part of the material in this paper was presented at IEEE PIMRC'19 \cite{AzaGerGar19}. The work of M.~M.~Azari was partly supported by the Catalan government under grant 2017 SGR1479. The work of G.~Geraci was partly supported by MINECO under Project RTI2018-101040-A-I00 and by the Postdoctoral Junior Leader Fellowship Programme from ``la Caixa" Banking Foundation.}
}

\maketitle
\thispagestyle{empty}
\begin{abstract}
We consider a cellular network deployment where UAV-to-UAV (U2U) transmit-receive pairs share the same spectrum with the uplink (UL) of cellular ground users (GUEs). For this setup, we focus on analyzing and comparing the performance of two spectrum sharing mechanisms: (i) underlay, where the same time-frequency  resources may be accessed by both UAVs and GUEs, resulting in mutual interference, and (ii) overlay, where the available resources are divided into orthogonal portions for U2U and GUE communications. We evaluate the coverage probability and rate of both link types and their interplay to identify the best spectrum sharing strategy. We do so through an analytical framework that embraces realistic height-dependent channel models, antenna patterns, and practical power control mechanisms. For the underlay, we find that although the presence of U2U direct communications may worsen the uplink performance of GUEs, such effect is limited as base stations receive the power-constrained UAV signals through their antenna sidelobes. In spite of this, our results lead us to conclude that in urban scenarios with a large number of UAV pairs, adopting an overlay spectrum sharing seems the most suitable approach for maintaining a minimum guaranteed rate for UAVs and a high GUE UL performance.
\end{abstract}
\IEEEpeerreviewmaketitle
\begin{IEEEkeywords}
UAV-to-UAV communications, D2D communications, cellular networks, spectrum sharing, stochastic geometry.
\end{IEEEkeywords}
\section{Introduction}
\label{sec:Intro}
The telecommunications industry and academia have long agreed on the social benefits that can be brought by having cellular-connected unmanned aerial vehicles (UAVs) \cite{AzaGerGar19,geraci2019preparing, Qualcomm2017,fotouhi2018survey, MozSaaBen2018, vinogradov2019tutorial}. These include facilitating search-and-rescue missions, acting as mobile small cells for providing coverage and capacity enhancements, or automating logistics in indoor warehouses \cite{Cbinsights:19,ZenZhaLim2016,azari2017ultra}. From a business standpoint, mobile network operators may benefit from offering cellular coverage to a heterogeneous population of terrestrial and aerial users \cite{Ericsson:18,YanLinLi18,AzaRosPol2017}.

\subsection{Motivation and Related Work}

A certain consensus has been reached---both at 3GPP meetings and in the classroom---on the fact that present-day networks will be able to support cellular-connected UAVs up to a certain extent \cite{azari2017coexistence,LopDinLi2018GC,NguAmoWig2018,ZenLyuZha2018,MeiWuZhang2018,azari2019cellular}. Besides, recent studies have shown that 5G-and-beyond hardware and software upgrades may be required by both mobile operators and UAV manufacturers to target large populations of UAVs flying at high altitudes \cite{GarGerLop2018,GerGarGal2018,ChaDanLar2017,DanGarGerICC2019}.

However, important use-cases exist where direct communication between UAVs, bypassing ground network infrastructure, would be a key enabler. These include autonomous flight of UAV swarms, collision avoidance, and UAV-to-UAV relaying, data transfer, and gathering \cite{zeng2019accessing,ZhaZhaDi19,FabCalCan17}. Similarly to ground device-to-device (D2D) communications \cite{LinAndGho2014,ChuCotDhi2017,GeoMunLoz2015,RimDar2017,AsaWanMan2014}, UAV-to-UAV (U2U) communications may also have implications in terms of spectral and energy efficiencies, extended cellular coverage, and reduced backhaul demands.

\subsection{Methodology and Contribution}

In this paper, we consider a cellular network deployment where UAV transmit-receive pairs share the same spectrum with the uplink (UL) of cellular ground users (GUEs). We examine two strategies for spectrum sharing, namely \emph{underlay} and \emph{overlay}. In the underlay, UAVs are allowed to access a fraction of the time-frequency physical resource blocks (PRBs) available for the GUE UL, resulting in mutual interference. In the overlay, the available PRBs are split into two orthogonal portions, respectively reserved for each link type.

Through stochastic geometry tools, we characterize the performance of U2U links and GUE UL, as well as their interplay, under both spectrum sharing mechanisms. Specifically, we evaluate the impact that the UAV altitude, UAV density, UAV power control, U2U link distance, and the number of PRBs accessed by each link type have on the coexistence of aerial and ground communications. To the best of our knowledge, this work is the first one to do so by accounting for: (i) a realistic, height-dependent propagation channel model, (ii) the impact of a practical base station (BS) antenna pattern, and (iii) a fractional power control policy implemented by all nodes.

Under such realistic setup, we first obtain exact analytical expressions for the coverage probability, i.e., the signal-to-interference-plus-noise ratio (SINR) distribution, of all links with both underlay and overlay approaches. As these expressions may require a considerable effort to be numerically evaluated, we also propose tight approximations based on practical assumptions. We validate both our exact and approximated analysis through simulations, and provide numerical results to gain insights into the behavior of U2U communications in cellular networks.

\subsection{Summary of Results}

Our main takeaways can be summarized as follows.
\begin{itemize}[leftmargin=*]
\item \textit{Link interplay:} In the underlay, the presence of U2U links may degrade the GUE UL. Such performance loss is limited by the fact that BSs perceive interfering UAVs through their antenna sidelobes, and UAVs can generally transmit at low power thanks to the favorable U2U channel conditions. However, the performance of both U2U and GUE UL links worsens as UAVs fly higher. This is due to an increased probability of line-of-sight (LoS)---and hence interference---on all UAV-to-UAV, GUE-to-UAV, and UAV-to-BS interfering links. Such negative effect outweighs the benefits brought by having larger GUE-to-UAV and UAV-to-BS distances.
\item \textit{Power control policy:} In the underlay, the UAV power control policy has a significant impact on all links. A tradeoff exists between the performance of U2U and GUE UL communications, whereby increasing the UAV transmission power improves the former at the expense of the latter. Moreover, smaller U2U distances can benefit both U2U and GUE UL links. Indeed, owed to the reduced path loss experienced by U2U pairs, UAVs may employ a smaller transmission power and therefore reduce the interference they cause to other U2U links and to GUEs.
\item \textit{Spectrum allocation:} In the underlay, where GUE-to-UAV interference is dominant, the rate degradation at UAVs caused by increasing their density is limited. However, increasing the number of PRBs utilized by U2U pairs causes a sharp performance degradation for GUEs, unless both the UAV density and the UAV transmission powers are limited. Implementing an overlay spectrum sharing approach may be the best option in order to maintain a high GUE UL performance while guaranteeing a minimum rate of 100~kbps to the majority of U2U pairs.
\end{itemize}

\subsection{Article Outline}

The remainder of this article is structured as follows. We introduce the system model in Section~II. In Section~III, we analyze the exact coverage probability of U2U and GUE UL links under underlay and overlay spectrum sharing. In Section~IV, we derive more compact, tight approximations for the coverage probability based on realistic assumptions. We show numerical results in Section~V to validate our analysis and approximations, and we provide several takeaways to the reader. We summarize our findings in Section~VI.


\section{System Model}
\label{sec:System_Model}

\begin{figure*}[!t]
\centering
\includegraphics[width=1.85\columnwidth]{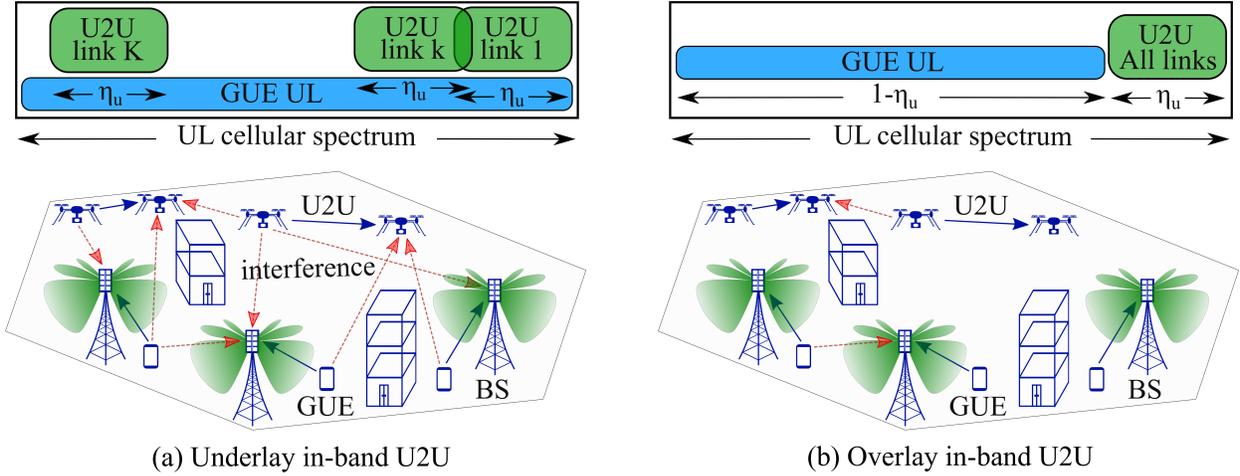}
	\caption{U2U communications sharing spectrum with the cellular UL. Blue solid (resp. red dashed) arrows indicate communication (resp. interfering) links. In (a)---underlay in-band U2U---GUEs occupy the whole spectrum while UAVs occupy a fraction $\eta_\u$, where mutual GUE-U2U interference occurs. In (b)---overlay in-band U2U---the spectrum is split into orthogonal portions, with a fraction $\eta_\u$ reserved to UAVs.}
	\label{U2U_SystemModel}
\end{figure*}

In this section, we introduce the network topology, channel model, spectrum sharing, and power control mechanisms considered throughout the paper. The main notations employed are summarized in Table~\ref{table:notation}, whereas further details on the parameters used in our study are provided in Table~\ref{table:parameters}.

\subsection{Network Topology}

We consider a cellular system as depicted in Fig.~\ref{U2U_SystemModel}, where (i) the UL transmissions of GUEs, and (ii) U2U transmit-receive pairs reuse the same spectrum. In the sequel, we employ the subscripts $\{\u,\c,\mathrm{b}\}$ to denote UAV, GUE, and BS nodes, respectively.

\subsubsection*{i) Ground UL cellular communications}
The BSs of the ground cellular network are deployed at a height $\hb$, are uniformly distributed as a Poisson point process (PPP) $\Phib \in \mathbb{R}^2$ with density $\lamb$, and communicate with their respective sets of connected GUEs. Assuming that the number of GUEs is sufficiently large when compared to that of the BSs, the active GUEs on each PRB form an independent Poisson point process $\Phig \in \mathbb{R}^2$ with density $\lamg=\lamb$ \cite{ChuCotDhi2017}. We further consider that GUEs associate to their closest BS, which generally also provides the largest reference signal received power (RSRP).\footnote{A GUE may connect to a BS $b$ other than the closest one $a$ if its link is in LoS with $b$ and not with $a$. However, since the probability of LoS decreases with the distance, such event is unlikely to occur \cite{3GPP36777}.} Therefore, the 2-D distance between a GUE and its associated BS follows a Rayleigh distribution with a scale parameter given by $\sigma_\c = 1/\sqrt{2\pi \lamc}$. When focusing on a typical BS serving its associated GUE, the interfering GUEs form a non-homogeneous PPP with density $\lamci(r) = \lamb(1-e^{-\lamb \pi r^2})$, where $r$ is the 2-D distance between the interfering GUE and the typical BS \cite{ChuCotDhi2017,SinZhaAnd:15,YanGerQue:16}.

\subsubsection*{ii) Direct UAV-to-UAV communications}
We consider that U2U transmitters form a PPP $\Phiu$ with intensity $\lamu$, and that each U2U receiver is randomly and independently placed around its associated transmitter with distance $\Ru$ distributed as $f_{\Ru}(\ru)$.

\subsection{Spectrum Sharing Mechanisms}


Let the available spectrum be divided into $\mathrm{n}$ PRBs. We consider the two spectrum sharing strategies---underlay and overlay--- illustrated in Fig.~\ref{U2U_SystemModel} and described as follows.

\subsubsection*{Underlay in-band U2U}

Each PRB may be used by both link types \cite{LinAndGho2014}. In particular, we assume that:
\begin{itemize}[leftmargin=*]
\item Each active GUE occupies all $\mathrm{n}$ PRBs. This is consistent with a cellular operator's goal of preserving the performance of its legacy ground users \cite{GarGerLop2018,GerGarGal2018}.
\item Each U2U transmitter occupies a fraction $\eta_\u$ of all PRBs, also employing frequency hopping to randomize its interference to other links. Specifically, each U2U transmitter may randomly and independently access $\eta_\u \cdot \mathrm{n}$ PRBs, where the factor $\eta_\u \in [0,1]$ measures the aggressiveness of the U2U spectrum access, and is denoted the spectrum access factor in the underlay. As a result, the density of interfering UAVs is given by $\hat{\lambda}_\u = \eta_\u \cdot \lamu$.
\end{itemize}

\subsubsection*{Overlay in-band U2U}

The available UL spectrum is split into two orthogonal portions. A fraction $\eta_\u$ is reserved for U2U communications, and UAVs access all $\eta_\u \cdot \mathrm{n}$ allocated PRBs without frequency hopping. Similarly, the remaining fraction $\eta_\g = 1 - \eta_\u$ is reserved to the GUEs UL, and active GUEs access all $\eta_\g \cdot \mathrm{n}$ PRBs allocated. This approach results in each GUE UL link being interfered only by other GUEs, and in each U2U link being interfered only by other UAVs.

In scenarios with the same number of UAVs, it is worth noting that UAVs will perceive more UAV-generated interference in the overlay when compared to the underlay, since all UAV pairs utilize the same PRBs. Accordingly, GUEs receive no interference from the UAVs in the overlay, at the expense of having to access only a subset of the available PRBs.

\subsection{Propagation Channel}

We assume that any radio link between nodes $\x$ and $\y$ is affected by large-scale fading $\zetaxy$, comprising path loss $\tauxy$ and antenna gain $\gxy$, and small-scale fading $\psixy$.

\subsubsection*{Probability of LoS}
We consider that links experience line-of-sight (LoS) and non-LoS (NLoS) propagation conditions with probabilities $\prxy^\L$ and $\prxy^\N$, respectively. In what follows, we make use of the superscripts $\nu,\xi \in \{\mathrm{L},\mathrm{N}\}$ to denote LoS and NLoS conditions on a certain link.


\subsubsection*{Path loss}
The distance-dependent path loss between two nodes $\mathrm{x}$ and $\mathrm{y}$ is given by
\begin{equation}
\tauxy = \hat{\tau}_\mathrm{xy} \, \dxy^{\,\alphaxy},
\end{equation}
where $\hat{\tau}_\mathrm{xy}$ denotes the reference path loss, $\alphaxy$ is the path loss exponent, and $\dxy = \sqrt{\rxy^2 + \hxy^2}$, $\rxy$, and $\hxy = |\hx - \hy|$ represent the 3-D distance, 2-D distance, and height difference between $\mathrm{x}$ and $\mathrm{y}$, respectively. Table~\ref{table:parameters} lists the path loss parameters employed in our study, which depend on the nature of $\x$ and $\y$.

\subsubsection*{Antenna gain}
We assume that all GUEs and UAVs are equipped with a single omnidirectional antenna with unitary gain. On the other hand, we consider a realistic BS antenna radiation pattern to capture the effect of sidelobes, which is of particular importance in UAV-to-BS links \cite{azari2017coexistence,GerGarGal2018}. We assume that each BS is equipped with a vertical, $\mathrm{N}$-element uniform linear array (ULA), where each element has directivity
\begin{equation} \label{ElementGain}
g_E(\theta) = g_E^{\max} \sin^2\theta
\end{equation}
as a function of the zenith angle $\theta$. The total BS radiation pattern $g_b(\theta) = g_E(\theta)\cdot g_A(\theta)$ is obtained as the superposition of each element's radiation pattern $g_E(\theta)$ and by accounting for the array factor given by
\begin{equation}
g_A(\theta) = \frac{\sin^2\Big(N\pi (\cos\theta - \cos\tt)/2\Big)}{N\sin^2\Big(\pi (\cos\theta - \cos\tt)/2\Big)},	
\end{equation}
where $\tt$ denotes the electrical downtilt angle.
The total antenna gain $\gxy$ between a pair of nodes $\x$ and $\y$ is given by the product of their respective antenna gains.

\subsubsection*{Small-scale fading}
On a given PRB, $\sixy$ denotes the small-scale fading power between nodes $\x$ and $\y$. Given the different propagation features of ground-to-ground, air-to-air, and air-to-ground links, we adopt the general Nakagami-m small-scale fading model. As a result, the cumulative distribution function (CDF) of $\sixy$ is given by
\begin{equation} \label{FadingCDF}
F_{\sixy}(\omega) \triangleq \mathbb{P}[\sixy < \omega] \!=\! 1\!-\!\sum_{i=0}^{\mxy-1} \!\frac{(\mxy \omega)^i}{i!} e^{-\mxy \omega},\!
\end{equation}
where $\mxy \in \mathbb{Z}^{+}$ is the fading parameter, with LoS links typically exhibiting a larger value of $\mxy$ than NLoS links.

\subsection{Power Control}

As per the cellular systems currently deployed, we consider fractional power control for all nodes. Accordingly, the power transmitted per PRB by a given node $\x$ is adjusted depending on the receiver $\y$ and can be computed as \cite{BarGalGar2018GC}
\begin{equation}
\Px = \min\left\{ \Pxmax, \rhox \cdot \zetaxy^{\epsx} \right\},
\label{eqn:power_control}
\end{equation}
where $\Pxmax$ is given by the maximum transmit power over the whole spectrum allocated to the node, divided by the number of PRBs utilized by node $\x$ for transmission, i.e., $P_{\mathrm{max}} / \mathrm{n_x}$. In (\ref{eqn:power_control}), $\rhox$ is a parameter adjusted by the network, $\epsx \in [0,1]$ is the fractional power control factor, and $\zetaxy = \tauxy /\gxy$ is the large-scale fading between nodes $\x$ and $\y$. The aim of (\ref{eqn:power_control}) is to compensate for a fraction $\epsx$ of the large-scale fading, up to a limit imposed by $\Pxmax$ \cite{3GPP36777}.

\subsection{Key Performance Indicators}

In what follows, we will analyze the coverage probability, denoted by $\mathcal{C}_\x$ for node $\x$. This is defined as the complementary CDF (CCDF) of the SINR, i.e., the probability of the SINR at node $\x$, SINR$_\x$, being beyond a certain threshold $\t$:
\begin{equation}
	\mathcal{C}_\x(\t) \triangleq \mathbb{P}\{\sinr_\x>\t\}.
\end{equation}
The rate $\Rx$ achievable by node $\x$ is related to its SINR as $\Rx = \Bx \log_2(1+\sinr_\x)$, with $\Bx$ denoting the bandwidth accessed by node $\x$. From the coverage probability, the coverage rate probability can be obtained as the CCDF of the achievable rate $\Rx$ at node $\x$ \cite{bai2015coverage}:
\begin{equation}
\mathbb{P}[\Rx>\t]=\Cx(2^{\t/\Bx}-1).
\end{equation}

%
\begin{table}
\centering
\caption{Notations.}
\label{table:notation}
\def\arraystretch{1.2}
\begin{tabulary}{\columnwidth}{ |p{1.7cm} | p{5.6cm} | }
\hline
	\textbf{Notation} 			 & \textbf{Definition} \\ \hline
$\lamb$ ($\lamu$)          & BS (UAV) density  \\\hline
$\bar{R}_\mathrm{u}$ ($\sigma_\u$) & mean (scale parameter) of U2U distance \\ \hline
$\rMuu$ & maximum U2U distance\\ \hline
$\prxy^\L$ ($\prxy^\N$)                                       & probability of LoS (NLoS) between x and y  \\\hline
$\nu,~\xi \in \{\mathrm{L},\mathrm{N}\}$ & superscripts denoting LoS or NLoS condition \\ \hline
$\alphaxyL$ ($\alphaxyN$)            & LoS (NLoS) path loss exponent for x--y link   \\\hline
$\psixyL$ ($\psixyN$)                & LoS (NLoS) small-scale fading for x--y link  \\\hline
$\mxyL$ ($\mxyN$)                    & LoS (NLoS) Nakagami-m parameter for x--y link   \\\hline
$\gxy$  & total antenna gain for x--y link \\\hline
$\hat{\tau}_\mathrm{xy}^\L$ ($\hat{\tau}_\mathrm{xy}^\N$) & LoS (NLoS) reference path loss\\ \hline
$\rxy$ ($\dxy$)                      & 2-D (3-D) distance for x--y link   \\\hline
$\hx$ ($\hxy$)  & height of node x (difference between $\hx$ and $\hy$)  \\\hline
$\pcovuav$ ($\pcovc$)      & U2U (GUE) coverage probability \\ \hline
$\mathrm{T}$                                       & SINR threshold  \\\hline
$\mathrm{B_t}$ ($\mathrm{n}$) & total bandwidth (number of PRBs) \\ \hline
$\mathrm{B_x}$ ($\eta_\x$) & bandwidth (spectrum allocation factor) for  x  \\ \hline
$\pu$ ($\pc$)              & UAV (GUE) transmit power    \\\hline
$\pur$ ($\pcr$)            & reference value for UAV (GUE) power control     \\\hline
$\eu$ ($\ec$)              & UAV (GUE) power control factor    \\\hline
$\tt$ ($N$)  & BS tilt angle (number of antenna elements) \\ \hline
$\Ixy$ & aggregate interference imposed by x on y \\ \hline
$\mathrm{N_0}$ & noise power \\ \hline
$\gamma(\cdot,\cdot)$                            & lower incomplete gamma function    \\\hline
$\Gamma(\cdot)$                                  & Gamma function   \\\hline
${_2}\mathrm{F}_1(\cdot,\cdot;\cdot;\cdot)$      & hypergeometric function   \\\hline
$\mathds{1}(\cdot)$ & indicator function \\ \hline
$\D_z^i$ & i-th derivative with respect to $z$ \\ \hline
\end{tabulary}
\end{table}
\section{Exact Performance Analysis}
\label{sec:analysis_exact}

Our U2U (resp. GUE UL) performance analysis is conducted for a typical BS (resp. UAV) receiver located at the origin.
In what follows, uppercase and lowercase letters are employed to respectively denote random variables and their realizations, e.g., $\Ru$ and $\ru$.

\subsection{Exact U2U Coverage Probability}

\subsubsection*{Underlay in-band U2U}

We now derive the U2U link coverage probability in the underlay.

\begin{Theorem} \label{U2Ucov_theorem}
	The underlay U2U coverage probability can be obtained as
	\begin{align} \label{eq:U2ULinkCoverage}
	\pcovuav(\t) = \sum_{\nu \in \{\mathrm{L},\mathrm{N}\}}\int_0^\rMuu f_{\Ru}^\nu(\ru) \mathcal{C}_{\mathrm{u}|\Ru}^\nu(\ru)  \mathrm{d}\ru.
	\end{align}
	In \eqref{eq:U2ULinkCoverage}, $\mathcal{C}_{\mathrm{u}|\Ru}^\nu(\ru)$ is the coverage probability of a U2U link given its distance $\Ru = \ru$ and the link condition $\nu$ (LoS or NLoS), which is obtained as
	\begin{align} \label{eq:U2ULinkCondCoverage}
	\mathcal{C}_{\mathrm{u}|\Ru}^\nu(\ru) = \sum_{i=0}^{\mathrm{m_{uu}^\nu}-1} (-1)^i \q_{\u,i}^\nu \cdot \D^i_{\s_\u} \left[\lapiu^\nu(\mathrm{s_u})\right],
	\end{align}
	where
	\begin{align}
	\q_{\u,i}^\nu &\triangleq \frac{e^{-\mathrm{N_0}\mathrm{s_u}}}{i!}\sum_{j=i}^{\mathrm{m_{uu}^\nu}-1} \frac{\mathrm{N_0}^{j-i}\mathrm{s_u}^j}{(j-i)!}, \\
	\!\!\mathrm{s_u} &\triangleq \frac{\mathrm{m_{uu}^\nu} \t}{\pu^\nu(\ru) \zuu^\nu(\ru)^{-1}}.
	\end{align}
	In \eqref{eq:U2ULinkCondCoverage}, $\iu$ is the aggregate interference at the UAV receiver caused by interfering UAVs and GUEs and is characterized by its Laplacian, obtained as $\lapiu^\nu(\mathrm{s_u}) = e^{\Lambda(\mathrm{s_u})}$ with
	\begin{equation}
	\Lambda(\mathrm{s_u}) \!=\! -2 \pi\!\left[ \lamuhat \!\!\!\!\sum_{\xi \in \{\mathrm{L},\mathrm{N}\}}\!\!\mathcal{I}_\mathrm{uu}^\mathrm{\xi}(\mathrm{s_u})  \!+\! \lamb \!\!\!\!\sum_{\xi \in \{\mathrm{L},\mathrm{N}\}}\mathcal{I}_\mathrm{gu}^\mathrm{\xi}(\mathrm{s_u}) \right]\!,
	\end{equation}
	where for $\xi \in \{\mathrm{L},\mathrm{N}\}$
	\begin{equation}
	\begin{aligned}
	\mathcal{I}_\mathrm{xy}^\xi &= \int_0^\infty \!\hspace{-0.2cm} f_{R_\mathrm{x}}^\mathrm{L}(x)\sum_{i = 1}^\infty \left[\prxy^{\xi}(\r_{i-1})\!-\!\prxy^{\xi}(\r_{i})\right] \underbrace{\Psi_\mathrm{xy}^\xi\left(\mathrm{s},\r_{i}\right)}_\text{at $P_\x = P_\mathrm{x}^\mathrm{L}$}  \!\mathrm{d}x
	\\
	&\!\!\!\!\!\!\!\!+\! \int_0^\infty \!\hspace{-0.2cm} f_{R_\mathrm{x}}^\mathrm{N}(x) \sum_{i = 1}^\infty \!\left[\prxy^{\xi}(\r_{i-1})\!-\!\prxy^{\xi}(\r_{i})\right] \!\underbrace{\Psi_\mathrm{xy}^\xi\left(\mathrm{s},\r_{i}\right)}_\text{at $P_\x = P_\mathrm{x}^\mathrm{N}$}  \!\mathrm{d}x.
	\end{aligned}\label{Ixyxi}
	\end{equation}
	
	In \eqref{Ixyxi}, $\prxy^{\xi}(\r_0) \triangleq 0$, and
	\begin{equation}
	\begin{aligned} \label{psixyxi_def}
	\Psi_\mathrm{xy}^\xi(\mathrm{s},\r) &\triangleq \frac{\r^2+\mathrm{h_{xy}^2}}{2}\left[1-\left(\frac{\m}{\m+\mu(\s,\r)}\right)^{\m}\right] \\
	&\!\!\!\!\!\!\!- \mathcal{K}(s,\r)
	\,{_2}F_1\left(1+\m,1-\beta;2-\beta;-\frac{\mu(\s,\r)}{\m}\right),
	\end{aligned}
	\end{equation}
	where ${_2}F_1(\cdot)$ is the Gauss hypergeometric function, $\m = \mathrm{m_{xy}^\xi}$, $\beta = \frac{2}{\alpha_\mathrm{xy}^\xi}$, $\mathrm{s} = \mathrm{s_y}\frac{\mathrm{g_{xy}}}{\hat{\tau}_{\mathrm{xy}}^\xi}$, and
	\begin{align} \label{eeq5}
	\mu(\s,\r) \triangleq \frac{\mathrm{s} P_\mathrm{x}}{(\r^2\!+\!\mathrm{h_{xy}^2})^{1/\beta}},~~ 	\mathcal{K}(s,\r) \triangleq \frac{\mathrm{s} P_\mathrm{x}}{2(1\!-\!\beta)\,(\r^2\!+\!\mathrm{h_{xy}^2})^{1/\beta-1}}.
	\end{align}
\end{Theorem}

\begin{proof}
	See Appendix \ref{U2Ucov_proof}.
\end{proof}



\begin{Remark} \label{DerGuideline}
In order to compute the coverage probability in \eqref{eq:U2ULinkCondCoverage}, one needs to calculate the derivatives of $\lapiu^\nu(\mathrm{s_u})$. Such derivation can be performed as explained in Appendix~\ref{Laplacian_derivatives}.
\end{Remark}

\subsubsection*{Overlay in-band U2U}

The overlay U2U coverage probability can be obtained by setting $\lamb = 0$ and $\lamuhat = \lamu$ in Theorem~\ref{U2Ucov_theorem}. In this case, UAVs only perceive interference generated by other UAVs, and hence one can write for the Laplacian of the aggregate interference in \eqref{eq:U2ULinkCondCoverage}
\begin{equation}
	\lapiu^\nu(\mathrm{s_u}) = e^{-2 \pi \lamu \sum_{\xi \in \{\mathrm{L},\mathrm{N}\}}\!\!\mathcal{I}_\mathrm{uu}^\mathrm{\xi}(\mathrm{s_u})  }.
\end{equation}

\subsection{Exact GUE UL Coverage Probability}

\subsubsection*{Underlay in-band U2U}

We now obtain the GUE UL coverage probability in the underlay, i.e., the CCDF of the UL SINR experienced by a GUE in the presence of U2U communications sharing the same spectrum.

\begin{Theorem} \label{C2Bcov_theorem}
	
	The underlay GUE UL coverage probability is given by	
	\begin{align}
	\pcovc(\t) &= \!\!\!\! \sum_{\nu \in \{\mathrm{L},\mathrm{N}\}}\int_0^\infty f_{R_\c}^\nu(\rc)\, \mathcal{C}_{\c|\Rc}^\nu(\rc)  \,\,\mathrm{d}\rc,
	\end{align}
	where $\mathcal{C}_{\c|\Rc}^\nu(\rc)$ is the GUE coverage probability given the distance to the typical BS, i.e., $\Rc = \rc$ and its condition $\nu$, i.e., LoS or NLoS, which can be expressed as
	\begin{align} \label{eq:GUEULCoverage}
		\mathcal{C}_{\c|\Rc}^\nu(\rc) &= \sum_{i=0}^{\mathrm{m_{gb}^\nu}-1} (-1)^i \q_{\c,i}^\nu \cdot \D^i_{\s_\c} \left[ \lapic^\nu(\s_\c)\right],
	\end{align}
	and where
	\begin{align}
	\q_{\c,i}^\nu &\triangleq \frac{e^{-\s_\c\mathrm{N_0}}}{i!}\sum_{j=i}^{\mathrm{m_{gb}^\nu}-1} \frac{\mathrm{N_0}^{j-i}\s_\c^j}{(j-i)!}, \\	\s_\c &\triangleq \frac{\mathrm{m_{gb}^\nu} \t}{\pc^\nu(\rc) \zcb^\nu(\rc)^{-1}}.
	\end{align}
	 In \eqref{eq:GUEULCoverage}, the interference is characterized by its Laplacian, which is obtained as
	\begin{align}
	\lapic &= e^{ -2 \pi \hat{\lambda}_\u \sum_{\xi \in \{\mathrm{L},\mathrm{N}\}}\mathcal{I}_\mathrm{ug}^\mathrm{\xi} } \cdot e^{ -(2 \pi \lamb)^2 \sum_{\xi \in \{\mathrm{L},\mathrm{N}\}}\mathcal{I}_\mathrm{gg}^\mathrm{\xi} },
	\end{align}
	where $\mathcal{I}_\mathrm{ug}^\mathrm{\xi}$ is
	\begin{align} \nonumber
	\mathcal{I}_\mathrm{ug}^\xi \!&= \!\!\! \int_0^\infty \!\hspace{-0.2cm} f_{\Ru}^\L(x) \!\sum_{i = 1}^\infty \pub^{\xi}(\r_i) \Big(\!\underbrace{\Psi_\mathrm{ub}^\xi\left(\mathrm{s},\r_{i+1}\right) \!-\! \Psi_\mathrm{ub}^\xi\left(\mathrm{s},\r_{i}\right)}_\text{at $P_\mathrm{u} = P_\mathrm{u}^\mathrm{L}$}\!\Big)  \mathrm{d}x
	\\ \label{Iubxi}
	&\!\!+\!\! \int_0^\infty \!\hspace{-0.2cm} f_{\Ru}^\mathrm{N}(x)\!\sum_{i = 1}^\infty \pub^{\xi}(\r_i) \Big(\!\underbrace{\Psi_\mathrm{ub}^\xi\left(\mathrm{s},\r_{i+1}\right) \!-\! \Psi_\mathrm{ub}^\xi\left(\mathrm{s},\r_{i}\right)}_\text{at $P_\mathrm{u} = P_\mathrm{u}^\mathrm{N}$}\!\Big)  \,\mathrm{d}x,
	\end{align}
	with $\mathrm{s} = \mathrm{s_{g}}\frac{\mathrm{g_{ub}(\r_i)}}{\hat{\tau}_{\mathrm{ub}}^\xi}$, whereas
	\begin{equation}
	\begin{aligned} \label{Iccxi}
	\mathcal{I}_\mathrm{gg}^\xi &\!\!=\!\! \int_0^\infty \!\prcb^\L(x) x e^{-\lamb \pi x^2} \\
	&\!\!\!\!\!\!\!\times \sum_{i = j(x)}^\mathrm{\infty} \!\!\prcb^{\xi}(\r_i)\, \!\Big(\!\underbrace{\Psi_\mathrm{gb}^\xi\left(\mathrm{s},\r_{i+1}\right) \!-\! \Psi_\mathrm{gb}^\xi\left(\mathrm{s},\r_{i}\right)}_{\text{at $P_{\g} = P_{\g}^\mathrm{L}$}} \!\Big) \mathrm{d}x \\
	&\!\!\!\!\!\!\!+ \int_0^\infty \prcb^\N(x) x e^{-\lamb \pi x^2} \\
	&\!\!\!\!\!\!\!\times \sum_{i = j(x)}^\mathrm{\infty} \prcb^{\xi}(\r_i)\, \Big(\underbrace{\Psi_\mathrm{gb}^\xi\left(\mathrm{s},\r_{i+1}\right) - \Psi_\mathrm{gb}^\xi\left(\mathrm{s},\r_{i}\right)}_{\text{at $P_\mathrm{g} = P_\mathrm{g}^\mathrm{N}$}} \Big) \mathrm{d}x,
	\end{aligned}
	\end{equation}
	with $\mathrm{s} = \mathrm{s_{g}}\frac{g_\mathrm{{gb}(\r_i)}}{\hat{\tau}_{\mathrm{gb}}^\xi}$. In \eqref{Iubxi} and \eqref{Iccxi}, $\Psi_\mathrm{ub}^\xi$ and $\Psi_\mathrm{gb}^\xi$ are obtained from \eqref{psixyxi_def}. In \eqref{Iccxi}, $j(x)$ is the index such that $x \in [r_{j(x)},r_{j(x)+1}]$ holds and we replace $r_{j(x)}$ with $x$ in the equation. 
\end{Theorem}


\begin{proof}
	See Appendix \ref{C2Bcov_proof}.
\end{proof}

\subsubsection*{Overlay in-band U2U}
The GUE coverage probability in the overlay is obtained by replacing $\hat{\lambda}_\u = 0$ in Theorem \ref{C2Bcov_theorem}. 
\section{Approximated Performance Analysis}
\label{sec:analysis_approx}


While exact, the expressions obtained in Section~\ref{sec:analysis_exact} for the coverage probability may require a considerable effort to be numerically evaluated, particularly for what concerns computing the derivatives of the Laplacian (see Appendix~\ref{Laplacian_derivatives}). In this section, we provide simpler, tight approximations based on practical
assumptions.





\subsection{Preliminaries}

In order to obtain more compact analytical expressions, we employ the following approximations whose accuracy will be validated in Section~\ref{sec:numerical}. 

\begin{Approximation} \label{FadingCDFapp}
	We approximate the CDF of the Nakagami-m small-scale fading power $\sixy$ in \eqref{FadingCDF} as
	\begin{equation} \label{eq:CDFapp}
	F_{\sixy}(\omega) \approx \left( 1- e^{-\mathrm{b_{xy}}\,\omega}\right)^{\mxy},
	\end{equation}
	where $\mathrm{b_{xy}}$ is a function of $\mxy$ provided in Table~\ref{table:bxy}.
	
	\begin{table}
		\centering
		\caption{Values of $\mathrm{b_{xy}}$ as a Function of $\mxy$.}
		\label{table:bxy}
		\def\arraystretch{1.2}
		\begin{tabulary}{\columnwidth}{|c|c|c|c|c|c|}
			\hline\hline
			$\mxy$ & 1 & 2 & 3 & 4 & 5  \\ \hline
			$\mathrm{b_{xy}}$ & 1 & 1.487 & 1.81 & 2.052 & 2.246 \\ \hline\hline\hline
			$\mxy$ & 6 & 7 & 8 & 9 & 10  \\ \hline
			$\mathrm{b_{xy}}$ & 2.408 & 2.546 & 2.668 & 2.775 & 2.872 \\ \hline\hline
		\end{tabulary}
	\end{table}
	
\end{Approximation}

Approximation~\ref{FadingCDFapp} is inspired by \cite{bai2015coverage} and allows to derive closed-form expressions for the Laplacian of the interference, and in turn for the coverage probability. The value of $\mathrm{b_{xy}}$ is obtained through curve fitting.


\begin{Approximation} \label{NLoSapp}
We neglect the interference caused by NLoS UAV-to-UAV, GUE-to-UAV, and UAV-to-BS links.
\end{Approximation}

Approximation~\ref{NLoSapp} holds due to a high probability of having LoS links dominating the interference \cite{azari2019cellular,GerGarGal2018,3GPP36777}.

\begin{Approximation} \label{MeanPowerapp}
We approximate the UAVs transmit power, which is a random variable, with its mean value.\footnote{The distribution of the UAV transmit power in turn depends on the probability of LoS between any pair of nodes. In Section~\ref{sec:numerical}, Proposition~\ref{proposition:meanUAVtxPower}, we calculate the mean UAV transmit power for the case where the probability of LoS follows the well known ITU model \cite{ITU2012}.}
\end{Approximation}	

Approximation~\ref{MeanPowerapp} is motivated by the fact that U2U links tend to undergo LoS conditions, and thus a lower path loss exponent \cite{3GPP36777}. This implies a lower variation of the UAV transmit power with respect to its distance from the receiver. This approximation removes one integral in the computation of the coverage probability. 


\subsection{Approximated U2U Coverage Probability} \label{sec:U2U_performance}

\subsubsection*{Underlay in-band U2U} We now make use of the aforementioned approximations to obtain a more compact form for the U2U coverage probability in the underlay.

\begin{Corollary} \label{corollary:U2Ucoverage}
	Under Approximations~1-3, the underlay U2U coverage probability is given by
	\begin{align} 
	\pcovuav(\t) = \int_0^\rMuu f_{\Ru}^\L(\ru) \mathcal{C}_{\mathrm{u}|\Ru}^\L(\ru) \mathrm{d}\ru,
	\end{align}
	where
	\begin{equation}
	\mathcal{C}_{\mathrm{u}|\Ru}^\L(\ru) = \sum_{i=1}^{\mathrm{m_{uu}^\L}} \binom{\mathrm{m_{uu}^\L}}{i}(-1)^{i+1} e^{-z_{\u,i}^\L \mathrm{N_0}} \cdot \lapiu^\L(z_{\u,i}^\L),
	\end{equation}
	\begin{align}
	\lapiu^\L(z_{\u,i}^\L) = \underbrace{e^{ -2 \pi \hat{\lambda}_\u \mathcal{I}_\mathrm{uu}^\L}}_{\text{due to LoS UAVs}} \cdot \underbrace{e^{ -2 \pi \lamb \mathcal{I}_\mathrm{gu}^\L }}_{\text{due to LoS GUEs}},
	\end{align}
	and
	\begin{align}
	\mathcal{I}_\mathrm{uu}^\L &= \sum_{j = 1}^\infty \left[\pruu^{\L}(\r_{j-1})-\pruu^{\L}(\r_{j})\right] \underbrace{\Psi_\mathrm{uu}^\L\left(\mathrm{s},\r_{j}\right)}_\text{at $\pu = \bar{P}_\mathrm{u}$},
	\end{align}
	with $\mathcal{I}_\mathrm{gu}^\L$ and $\Psi_\mathrm{uu}^\L$ defined in Theorem \ref{U2Ucov_theorem} and
	\begin{equation}
	\mathrm{s} = z_{\u,i}^\L \frac{\mathrm{g_{uu}}}{\hat{\tau}_{\mathrm{uu}}^\L};
	~~z_{\u,i}^\L = \frac{i b_{\u\u}^\L \t}{\pu^\L \zuu^\L(\ru)^{-1}}.
	\end{equation} 
\end{Corollary}

\begin{proof}
See Appendix \ref{proof:U2Ucoverage_approx}.
\end{proof}

\subsubsection*{Overlay in-band U2U}
Under Approximations~1-3, the overlay U2U coverage probability can be obtained from Corollary~\ref{corollary:U2Ucoverage} by substituting $\lamb=0$, $\lamuhat=\lamu$, and
	\begin{equation}
	\lapiu^\L(z_{\u,i}^\L) = e^{-2 \pi \lamu \mathcal{I}_\mathrm{uu}^\mathrm{\L}(z_{\u,i}^\L)  },
	\end{equation}
since the aggregate interference only includes the UAV-generated one.



\subsection{Approximated GUE UL Coverage Probability} \label{sec:GUE_performance}

\subsubsection*{Underlay in-band U2U}
Similarly, we now make use of the proposed approximations to obtain a more compact form for the GUE UL coverage probability in the underlay.

\begin{Corollary} \label{proposition:GUEcov}
Under Approximations~1-3, the underlay GUE UL coverage probability is given by
	\begin{align} 
	\pcovc(\t) &= \!\!\!\! \sum_{\nu \in \{\mathrm{L},\mathrm{N}\}}\int_0^\infty f_{R_\c}^\nu(\rc)\, \mathcal{C}_{\c|\Rc}^\nu(\rc)  \,\,\mathrm{d}\rc,
	\end{align}
	where
	\begin{align} 
	\mathcal{C}_{\c|\Rc}^\nu(\rc) &= \sum_{i=1}^{\mathrm{m_{gb}^\nu}} \binom{\mathrm{m_{gb}^\nu}}{i}(-1)^{i+1} e^{-z_{\g,i}^\nu \mathrm{N_0}} \cdot \lapic^\nu(z_{\g,i}^\nu),
	\end{align}
	and  
	\begin{align} \label{eqn:Lig}
	\lapic^\nu(z_{\g,i}^\nu) &= \underbrace{e^{ -2 \pi \hat{\lambda}_\u \mathcal{I}_\mathrm{ug}^\L }}_{\text{due to LoS UAVs}} \cdot \underbrace{e^{ -2 \pi \lamb \sum_{\xi \in \{\mathrm{L},\mathrm{N}\}}\mathcal{I}_\mathrm{gg}^\mathrm{\xi} }}_{\text{due to GUEs}}.
	\end{align}
	In (\ref{eqn:Lig}), $\mathcal{I}_\mathrm{ug}^\L$ is given by
	\begin{align}
	\mathcal{I}_\mathrm{ug}^\L &= \sum_{j = 1}^\infty \pub^{\L}(\r_j) \Big(\underbrace{\Psi_\mathrm{ub}^\L\left(\mathrm{s},\r_{j+1}\right) - \Psi_\mathrm{ub}^\L\left(\mathrm{s},\r_{j}\right)}_\text{at $\pu = \bar{P}_\mathrm{u}$}\Big),
	\end{align}
	whereas $\mathcal{I}_\mathrm{gg}^\mathrm{\xi}$ and $\Psi_\mathrm{ub}^\L$ are provided in Theorem \ref{C2Bcov_theorem} where we replace $\mathrm{s_g}$ with
	\begin{equation}
	z_{\g,i}^\nu = \frac{i b_{\c\mathrm{b}}^\nu \t}{\pc^\nu \zcb^\nu(\ru)^{-1}}.
	\end{equation}
\end{Corollary}

\begin{proof}
Similar	to proof of Corollary~\ref{corollary:U2Ucoverage} and thus omitted.
\end{proof}

\subsubsection*{Overlay in-band U2U}
Under Approximations~1-3, the overlay GUE UL coverage probability can be obtained from Corollary~\ref{proposition:GUEcov} by replacing $\hat{\lambda}_\u = 0$, since the aggregate interference only includes the GUE-generated one.

\section{Numerical Results and Discussion}
\label{sec:numerical}

In this section, we first validate our analysis and then characterize the performance of U2U and GUE UL cellular communications under overlay and underlay spectrum sharing strategies. Specifically, we consider an urban scenario, and we concentrate on evaluating how aerial and ground communications are affected by the UAV altitude, density, power control, link distance, and resource utilization. 
Unless otherwise specified, the system parameters are included in Table~\ref{table:parameters} and follow the 3GPP specifications in \cite{3GPP36777}.

\begin{table}
	\centering
	\caption{System Parameters}
	\label{table:parameters}
	\def\arraystretch{1.2}
	\begin{tabulary}{\columnwidth}{ |p{2.2cm} | p{5.2cm} | }
		\hline
		\textbf{Deployment} 			&  \\ \hline  
		BS distribution		& PPP with $\lamb = 5$~/~Km$^2$, $\hb=25$~m\\ \hline
		GUE distribution 				& One active GUE per cell, $\hc=1.5$~m \\ \hline
		UAV distribution 				& $\lamu\!=\!1$\,/\,Km$^2$, $\bar{R}_\mathrm{u}\!=\!100\,$m, $\hu\!=\!100\,$m \\ \hline\hline 
		\textbf{Channel model} 			&  \\ \hline
		\multirow{2}{*}{Ref. path loss [dB]}		&  $\hat{\tau}_\mathrm{gb}^\L = 28+20\log_{10}(f_c)$ \enspace ($f_c$ in GHz) \\ \cline{2-2}
		& $\hat{\tau}_\mathrm{gb}^\N = 13.54+20\log_{10}(f_c)$ \\ \cline{2-2}
		& $\hat{\tau}_\mathrm{ub}^\L = 28+20\log_{10}(f_c)$ \\ \cline{2-2}
		& $\hat{\tau}_\mathrm{ub}^\N = -17.5+20\log_{10}(40\pi f_c/3)$ \\ \cline{2-2}
		& $\hat{\tau}_\mathrm{gu}^\L = 30.9+20\log_{10}(f_c)$  \\ \cline{2-2}
		&  $\hat{\tau}_\mathrm{gu}^\N = 32.4+20\log_{10}(f_c)$  \\ \cline{2-2}
		& $\hat{\tau}_\mathrm{uu}^\L = 28+20\log_{10}(f_c)$\\ \cline{2-2}
		& $\hat{\tau}_\mathrm{uu}^\N = -17.5+20\log_{10}(40\pi f_c/3)$  \\ \hline
		\multirow{2}{*}{Path loss exponent}		&  $\alcb = 2.2,~~~\ancb = 3.9$ \\ \cline{2-2}
		& $\aubl = 2.2,~~~\aubn = 4.6-0.7\log_{10}(\hu)$ \\ \cline{2-2}
		& $\acul = 2.225-0.05\log_{10}(\hu)$ \\ & $\acun = 4.32-0.76\log_{10}(\hu)$ \\ \cline{2-2}
		& $\auul = 2.2,~~~\auun = 4.6-0.7\log_{10}(\hu)$\\ \hline
		Small-scale fading  & Nakagami-m with $\mxy^\xi = 1$ for NLoS links, $\mxy^\xi = 3$ for LoS GUE links, and $\mxy^\xi = 5$ for LoS UAV links \\ \hline
		Prob. of LoS & ITU model as per \eqref{PrLoS} \\ \hline
		Thermal noise 				& -174 dBm/Hz with 7~dB noise figure \\ \hline \hline 
		\textbf{PHY} 			&  \\ \hline
		\multirow{2}{*}{Spectrum}		& Carrier frequency: 2~GHz \\ \cline{2-2} 
		& Bandwidth: 10 MHz with 50 PRBs \\ \hline 
		BS array configuration		& $8\times 1$ vertical, 1 RF chain, downtilt: $102^{\circ}$, element gain as in \eqref{ElementGain}, spacing: $0.5~\lambda$\\ \hline
		Power control		& Fractional, based on GUE-to-BS (resp. U2U) large-scale fading for GUEs (resp. UAVs), with $\ec = \eu = 0.6$, $\pcr = \pur = -58$~dBm, and $P^{\textrm{max}}_\c = P^{\textrm{max}}_\u = 24$~dBm \cite{BarGalGar2018GC}\\ \hline
		GUE/UAV antenna 		& Omnidirectional with 0~dBi gain \\ \hline 
		\end{tabulary}
\end{table}


\subsection{Preliminaries}


While our analysis holds for any transmit/receive UAV height, in the following we assume all UAVs to be located at the same height $\hu$, to evaluate the impact of such parameter.

We model the U2U link distance $\Ru$ via a truncated Rayleigh distribution with probability density function (PDF)
\begin{align}\label{U2U_distance_pdf}
f_{\Ru}(\ru) = \frac{\ru e^{-\r_\u^2/(2\sigma_\mathrm{u}^2)}}{\sigma_\mathrm{u}^2\left(1-\e^{-r_\mathrm{M}^2/(2\sigma_\mathrm{u}^2)}\right)}\cdot\mathds{1}(\ru<r_\mathrm{M}),
\end{align}
where $\r_\mathrm{M}$ is the maximum U2U link distance, $\mathds{1}(\cdot)$ is the indicator function, and $\sigma_\mathrm{u}$ is the Rayleigh scale parameter, related to the mean distance $\bar{R}_\mathrm{u}$ through $\sigma_\mathrm{u} = \sqrt{\frac{2}{\pi}}\bar{R}_\mathrm{u}$.

As for the probability of LoS between any pair of nodes x and y, we employ the well known ITU model \cite{ITU2012,azari2019cellular}:
\begin{equation}
\prxy^\L(r) \!=\! \!\!\!\!\!\!\! \prod_{j=0}^{\lfloor{\frac{r\sqrt{\a_1\a_2}}{1000}-1\rfloor}} \!\!\left[1\!-\!\exp\!\left(\!-\frac{\left[\!\hx\!-\!\frac{(j+0.5)(\hx-\hy)}{\mathrm{k}+1} \!\right]^2}{2 \a_3^2}\right) \!\right]\!,\!\!
\label{PrLoS}
\end{equation}
where $\{\a_1, \a_2, \a_3\}$ are environment-dependent parameters set to $\{0.3, 500, 20\}$ to model an urban scenario. The probability of NLoS is simply obtained as $\prxy^\N(r) = 1-\prxy^\L(r)$.

Employing (\ref{U2U_distance_pdf}) and (\ref{PrLoS}), the mean UAV transmit power is then obtained as follows.

\begin{Proposition} \label{proposition:meanUAVtxPower}
	The mean UAV transmit power is given by
	\begin{equation} \label{eqn:meanPower}
	\begin{aligned}
	\mathbb{E}[\pu] =  \!\!\! \sum_{\nu \in \{\mathrm{L},\mathrm{N}\}}  &\Bigg[\sum_{i=1}^{j} [C_i^\nu-C_{i+1}^\nu] \,\, \gamma\left(1+\frac{\auu^\nu \eu}{2},y_{i+1}\right)\\
	&+\sum_{i=j+1}^{k+1} [B_i^\nu-B_{i-1}^\nu] \, e^{-y_i} \Bigg],
	\end{aligned}
	\end{equation}
	where
	\begin{align} 
	C_i^\nu &\!=\! \frac{(2\sigma_\u^2)^{{\auu^\nu\eu/2}}\pur\left(\hat{\tau}_{\mathrm{uu}}^\nu/\guu\right)^{\eu}}{1-\e^{-r_\mathrm{M}^2/(2\sigma_\mathrm{u}^2)}} \cdot \pruu^\nu(r_i), \text{ for } i>0,
	\end{align}
	$B_j^\nu = 0$, $B_{k+1}^\nu = 0$, and
	\begin{align}
	B_i^\nu &= \frac{\pumax \, \pruu^\nu(r_i) }{1-\e^{-r_\mathrm{M}^2/(2\sigma_\mathrm{u}^2)}};~i>j,
	\end{align}
	and where $j = \lfloor{\frac{r_m^\nu\sqrt{\a_1\a_2}}{1000}\rfloor}+1$, $k = \lfloor{\frac{r_M\sqrt{\a_1\a_2}}{1000}\rfloor}+2$, $y_i = \frac{r_i^2}{2\sigma_u^2}$, $r_k = r_M$, and $r_{j+1} = r_m^\nu$. The latter is the distance at which the UAV reaches its maximum allowed transmit power, which depends on the link condition (LoS vs. NLoS) and can be obtained from \eqref{eqn:power_control} as follows
	\begin{equation}
	r_m^\nu = \left(\frac{\guu}{\hat{\tau}_{\mathrm{uu}}^\nu}\right)^{1/\auu^\nu} \cdot \left(\frac{\pumax}{\pur}\right)^{1/(\auu^\nu\eu)}.
	\end{equation}
\end{Proposition}

\begin{proof}
	See Appendix \ref{proof:meanUAVtxPower}. 
\end{proof}



\subsection{Analysis Validation and Impact of UAV Height}


Fig.~\ref{Validation} shows the coverage probability for GUE UL and U2U links in the underlay, with $\eta_\u=1$, obtained in three different ways: (i) with our approximated analysis in Section~\ref{sec:analysis_approx}, (ii) through our exact analysis in Section~\ref{sec:analysis_exact}, and (iii) via simulations. The three sets of curves exhibit a close match, thus validating our analysis for the underlay and---as a special case---for the overlay too.

\begin{figure}[t!]
	\centering
	\includegraphics[width=\figwidth]{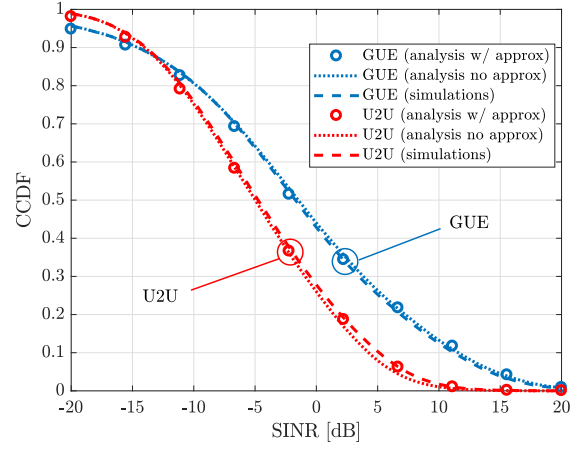}
	\caption{Underlay coverage probability obtained via approximated analysis (solid), exact analysis (dotted), and simulations (dashed).}
	\label{Validation}
\end{figure}


Fig.~\ref{Pcov_T_hu} shows the CCDF of the SINR per PRB in the underlay, with $\eta_\u=1$, experienced by: (i) U2U links, (ii) the UL of GUEs in the presence of U2U links, and (iii) the UL of GUEs without any U2U links. For (i) and (ii), we consider two UAV heights, namely 50~m and 150~m. 
In this figure, markers denote values obtained through our approximated expressions derived in Section~IV, whereas solid/dashed/dotted lines are obtained via simulations. Again, all curves show a close match, thus validating our analysis. Fig.~\ref{Pcov_T_hu} also allows to make a number of important observations:
\begin{itemize}[leftmargin=*]
\item U2U communications degrade the UL performance of GUEs. However, for the scenario where the UAVs fly at 50~m, such performance loss amounts to less than 3~dB in median, since (i) BSs perceive interfering UAVs through their antenna sidelobes, and (ii) UAVs generally transmit with low power due to the good U2U channel conditions.
\item The U2U performance degrades as UAVs fly higher, due to an increased UAV-to-UAV and GUE-to-UAV interference. The former is caused by a higher probability of LoS between a receiving UAV and interfering UAVs. The latter is caused by a higher probability of LoS between a receiving UAV and interfering GUEs, whose effect outweighs having larger GUE-UAV distances.
\item The GUE UL performance also degrades as UAVs fly higher. However, this degradation is less significant than that experienced by the U2U links, since interference generated by GUEs in other cells is dominant.
\end{itemize}

\begin{figure}
	\centering
	\includegraphics[width=\figwidth]{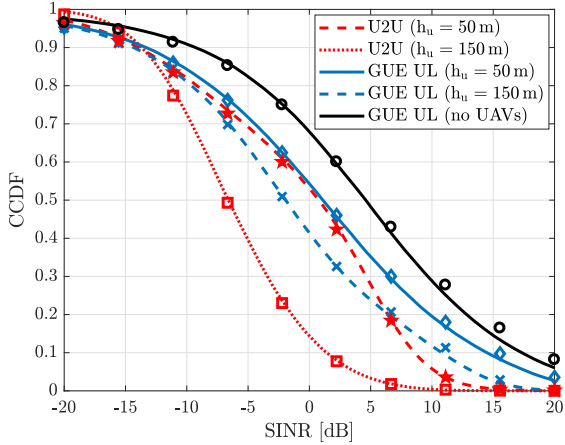}
	\caption{CCDF of the SINR per PRB experienced by: (i) U2U links, (ii) GUE UL in the presence of U2U links, and (iii) GUE UL without U2U links, in the underlay and for $h_{\u} = \lbrace 50, 150 \rbrace$~m. Curves and markers are respectively obtained via simulations and through our approximated analysis in Section~IV.}
	\label{Pcov_T_hu}
\end{figure}

After having validated their accuracy, in the remainder of this section we will use the expressions obtained through our approximated analysis in Section~IV.




\subsection{Effect of Power Control and Resource Allocation}


Fig.~\ref{Pcov_epsU_T} shows the probability of experiencing SINRs per PRB larger than -5~dB for both U2U and GUE UL in the underlay, with $\eta_\u=1$, as a function of $\eu$. We also consider three different values for the mean U2U distances $\bar{R}_\mathrm{u}$, namely 50~m, 100~m, and 150~m. Fig.~\ref{Pcov_epsU_T} allows us to draw the following conclusions:
\begin{itemize}[leftmargin=*]
\item The UAV power control policy has a significant impact on the performance of both U2U and GUE UL.\footnote{We refer the interested reader to \cite[Fig.~3]{AzaGerGar19} for a detailed breakdown of (i) the mean useful received power, (ii) the mean interference power received from GUEs, and (iii) the mean interference power received from UAVs, for both U2U and GUE UL links, as a function of the UAV fractional power control factor $\eu$.} There exists an inherent tradeoff, whereby increasing $\eu$ improves the former at the expense of the latter:
\begin{itemize}[leftmargin=*]
\item For $0 < \eu < 0.4$, the U2U performance is deficient, since UAVs use a very low transmission power. In this range, the GUE UL performance is approximately constant, since the GUE-generated interference is dominant.
\item For $0.4 < \eu < 0.9$, the U2U performance increases at the expense of the GUE UL.
\item For $\eu > 0.9$, the U2U performance saturates and that of the GUEs stabilizes, since almost all aerial devices reach their maximum transmit power.
\end{itemize}
\item Smaller U2U link distances---for fixed UAV density---correspond to a better U2U performance for all values of $\eu$. 
This is because (i) UAVs perceive larger received signal powers for decreasing $\bar{R}_\mathrm{u}$, since the path loss of the U2U links diminishes faster than the UAV transmit power when $\bar{R}_\mathrm{u}$ lessens, and (ii) the reduced UAV-to-UAV interference due to the smaller transmission power employed by UAVs.
\item The GUE UL also benefits from smaller U2U link distances when $\eu > 0.4$, since UAVs lower their transmit power and therefore reduce the UAV-to-BS interference.
\end{itemize}



\begin{figure}
	\centering
	\includegraphics[width=\figwidth]{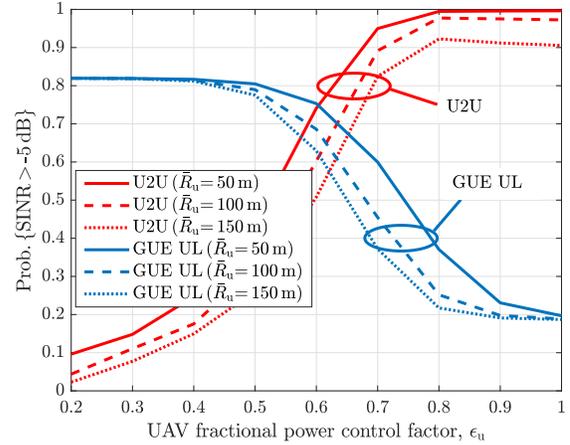}
	\caption{Probability of having SINRs $>-5$~dB for U2U and GUE UL links in the underlay vs. the UAV fractional power control factor $\eu$, and for $\bar{R}_\mathrm{u} = \lbrace 50, 100, 150 \rbrace$.}
	\label{Pcov_epsU_T}
\end{figure}


Fig.~\ref{Pcov_nPRBu_underlay_combined} shows the probability of experiencing SINRs per PRB larger than -5~dB for the GUE UL and U2U links in the underlay. We consider four configurations of the UAV fractional power control factor and spectrum access factor, i.e., $\eu=\{0.6,0.8\}$ and $\eta_{\mathrm{u}}=\{0.1,0.5\}$, and two values of the UAV density, i.e., $\lambda_{\mathrm{u}} = \{1\text{e-6},5\text{e-6}\}$, corresponding to red and blue markers, respectively. Notably, the results of Fig.~\ref{Pcov_nPRBu_underlay_combined} demonstrate how increasing $\eta_{\mathrm{u}}$, i.e., the number of PRBs allocated to UAV pairs, causes a sharp performance degradation for GUEs, except for the case where both the UAV density and the UAV transmit powers are constrained ($\lambda_{\mathrm{u}}=1$e-6, $\epsilon_{\mathrm{u}}=0.6$). As expected, also increasing the UAV density or transmit power generates more interference to the GUE UL, reducing the SINR. As for the U2U link performance, this remains almost constant with respect to $\eta_{\mathrm{u}}$ for $\lambda_{\mathrm{u}}=1$e-6, when UAV-to-UAV interference is negligible, whereas it decreases for $\lambda_{\mathrm{u}}=5$e-6, when UAV-to-UAV interference is more pronounced.



\begin{figure}[t!]
	\centering
	\includegraphics[width=\figwidth]{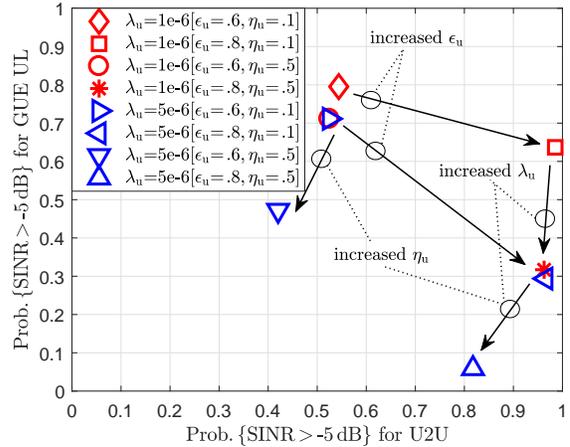}
	\caption{Probability of having SINRs $>-5$~dB for U2U and GUE UL links in the underlay for various combinations of $\epsilon_\u$, $\eta_\u$, and $\lambda_\u$.}
	\label{Pcov_nPRBu_underlay_combined}
\end{figure}

\subsection{Coverage Rate Comparison: Underlay vs. Overlay}


Fig.~\ref{CovRateU2U_combined} and Fig.~\ref{CovRateGUE_combined} show the CCDF of the coverage rate for U2U links and GUE UL, respectively, when $\eta_\u=0.1$, i.e., UAVs access five PRBs out of 50, in the underlay or in the overlay. 

Fig.~\ref{CovRateU2U_combined} provides the following insights:
\begin{itemize}[leftmargin=*]
\item In the overlay, the U2U coverage rate is only affected by UAV-to-UAV interference. Higher UAV densities thus have a more noticeable impact on the coverage rates than the UL power control strategy does. This can be observed by comparing scenarios with $\lambda_{\mathrm{u}}=1$e-6 (circled dotted red and circled dash-dotted purple curves) to scenarios with $\lambda_{\mathrm{u}}=5$e-6 (resp. circled solid green and circled dashed blue curves).
\item In the underlay, the U2U coverage rate is mostly affected by GUE-generated interference. Indeed, the rate degradation caused by increasing $\lambda_{\mathrm{u}}$ from 1e-6 to 5e-6 is limited when $\epsilon_{\mathrm{u}}=0.8$ (thick dash-dotted purple vs. dashed blue curves) and almost negligible when $\epsilon_{\mathrm{u}}=0.6$ (thick dotted red vs. solid green curves).
\item Comparing underlay vs. overlay, a crossover can be observed between green solid lines ($\epsilon_{\mathrm{u}}=0.6$, $\lambda_{\mathrm{u}}=5$e-6). This can be explained as follows. The upper part of the underlay CCDF corresponds to the worst U2U links---severely interfered by GUEs---which are better off in the overlay, where such interference is not present. The lower part of the underlay CCDF corresponds to the best U2U links---those not severely interfered by GUEs, for which UAV-to-UAV interference is dominant---that are worse off in the overlay, where all UAV interferers are concentrated on each PRB.
\end{itemize}

On the other hand, Fig.~\ref{CovRateGUE_combined} demonstrates that in order to maintain a high GUE UL rate, one should (i) adopt an overlay spectrum sharing approach, or (ii) limit the power employed by the UAVs in the underlay, i.e., set $\epsilon_{\mathrm{u}}=0.6$. However, we may also see from Fig.~\ref{CovRateU2U_combined} that setting $\epsilon_{\mathrm{u}}=0.6$ strongly reduces the U2U rates---almost by one order of magnitude in median for both $\lambda_{\mathrm{u}}=1$e-6 (thick dash-dotted purple vs. thick dotted red curves) and $\lambda_{\mathrm{u}}=5$e-6 (thick dashed blue vs. thick solid green curves).



\begin{figure}
	\centering
	\includegraphics[width=\figwidth]{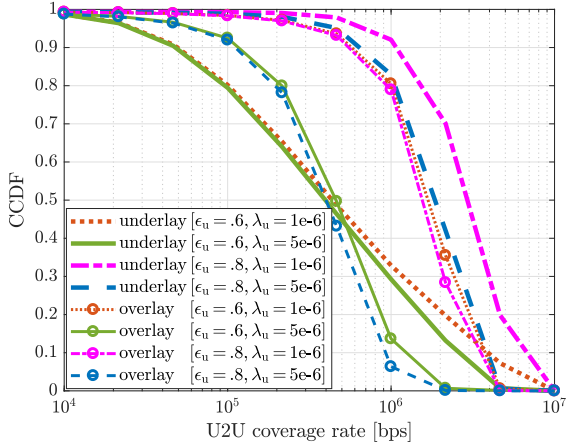}
	\caption{Coverage rate for U2U links with underlay and overlay, for various values of $\epsilon_{\mathrm{u}}$ and $\lambda_{\mathrm{u}}$, with UAVs accessing five PRBs ($\eta_\u = 0.1$).}
	\label{CovRateU2U_combined}
\end{figure}

\begin{figure}
	\centering
	\includegraphics[width=\figwidth]{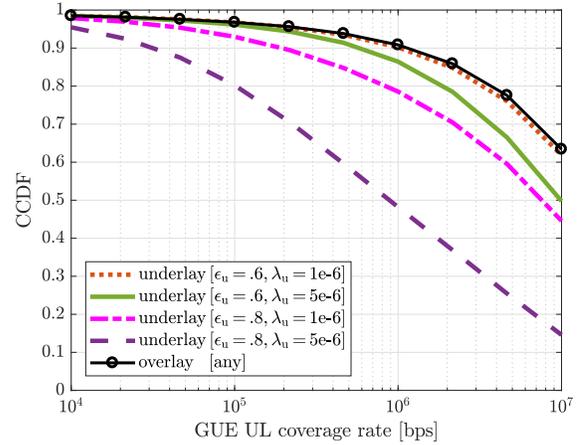}
	\caption{Coverage rate for GUE UL with underlay and overlay, for various values of $\epsilon_{\mathrm{u}}$ and $\lambda_{\mathrm{u}}$, with UAVs accessing five PRBs ($\eta_\u = 0.1$).}
	\label{CovRateGUE_combined}
\end{figure}

For ease of interpretation, Fig.~\ref{CovRate_tradeoff} combines Fig.~\ref{CovRateU2U_combined} and Fig.~\ref{CovRateGUE_combined}, illustrating the tradeoff between (i) the probability that U2U achieve rates of less than 100 kbps---a requirement set by the 3GPP for command and control information exchange \cite{3GPP36777}---, and (ii) the rates achieved by the $5\%$-worst GUEs. 
We consider two cases for the UAV density, namely $\lambda_{\mathrm{u}} = \{1\text{e-6},5\text{e-6}\}$, and four combinations for the spectrum sharing approach, namely \{underlay, overlay\} and $\eu=\{0.6,0.8\}$. We can observe from Fig.~\ref{CovRate_tradeoff} that, for both values of the UAV density $\lambda_{\mathrm{u}}$, the overlay spectrum sharing approach is capable of offering the best guaranteed GUE UL performance, while generally allowing a larger number of UAVs to achieve rates of 100~kbps.

\begin{figure}
	\centering
	\includegraphics[width=\figwidth]{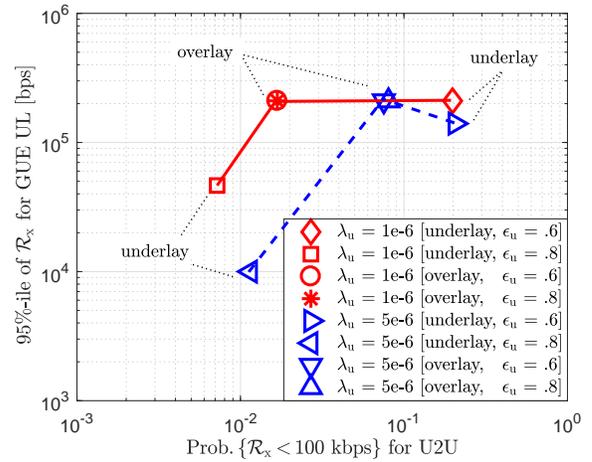}
	\caption{Tradeoff between (i) satisfying a requirement rate of 100~kbps for U2U links, and (ii) achieving a large rate for the $5\%$-worst GUEs with underlay and overlay, for $\lambda_{\mathrm{u}} = \left\{ \textrm{1e-6}, \textrm{5e-6} \right\}$ and $\epsilon_{\mathrm{u}} = \left\{ \mathrm{0.6, 0.8} \right\}$. }
	\label{CovRate_tradeoff}
\end{figure}
\section{Conclusion}
\label{sec:conclusion}

In this article, we provided an analytical framework to evaluate the performance of an uplink cellular network with both underlayed and overlayed U2U communications, while considering a realistic channel model, antenna pattern, and power control policy. 
In particular, we first derived exact analytical expressions for the coverage probability of all nodes, and then proposed practical assumptions that yield tight and compact approximations.

We found that in the underlay, (i) communications between pairs of close-by UAVs do not have a dramatic effect on the GUE UL---since the strong U2U channel gains allow UAVs to lower their transmit power---, and (ii) the U2U rate degradation caused by increasing the UAV density is limited---since the interference on U2U links is dominated by GUE transmissions. Instead, higher UAV densities result in lower U2U rates in the overlay, owing to all UAVs sharing the same resources without frequency hopping.

All in all, our results showed that overlaying U2U and GUE UL communications may be the preferable alternative in an urban scenario for simultaneously (i) maximizing the GUE UL performance, and (ii) guaranteeing a minimum U2U coverage rate of 100 kbps to the majority of UAV pairs.
\begin{appendix}

\subsection{Proof of Theorem \ref{U2Ucov_theorem}} \label{U2Ucov_proof}

	To obtain the U2U coverage probability, we can write
	
	\begin{equation}
	\begin{aligned}
	\pcovuav
	&= \mathbb{P}\left[\frac{\pu \zuu^{-1} \, \suu}{\mathrm{N_0} + I_\mathrm{u}} > \t \right]\\
	& = \sum_{\nu \in \{\mathrm{L},\mathrm{N}\}}\int_0^{\rMuu}  \mathcal{C}_{\mathrm{u}|\Ru}^\nu(\ru)\,f_{\Ru}^\nu(\ru) \,\mathrm{d}\ru,
	\end{aligned}
	\end{equation}
	where $f_{\Ru}^\nu(\ru) = f_{\Ru}(\ru) \cdot \pruu^\nu(\ru)$ and
	\begin{align} 
	\mathcal{C}_{\mathrm{u}|\Ru}^\nu(\ru) &\triangleq \mathbb{P}\left[\frac{\pu^\nu \zuu^\nu(\ru)^{-1} \, \suu^\nu}{\mathrm{N_0} + I_\mathrm{u}} > \t \right]. \label{CuvRuru}
	\end{align}
	$\mathcal{C}_{\mathrm{u}|\Ru}^\nu(\ru)$ can be obtained as follows
	
	\begin{align} \label{CuRu} \nonumber
	&\mathcal{C}_{\mathrm{u}|\Ru}^\nu(\ru) 
	= \mathbb{E}_{I_\mathrm{u}} \left\{\mathbb{P}\left[\suu^\nu > \frac{\t}{\pu^\nu \zuu^\nu(\ru)^{-1}}(\mathrm{N_0} + I_\mathrm{u}) \right] \right\} \\ \nonumber
	&\stackrel{\textrm{(a)}}{=} \mathbb{E}_{\iu} \left\{\sum_{i=0}^{\mathrm{m_{uu}^\nu}-1} \frac{\mathrm{s_u}^i}{i!} \, (\mathrm{N_0} + \iu)^i e^{-\mathrm{s_u} (\mathrm{N_0} + \iu)} \right\} \\ \nonumber
	&= \mathbb{E}_{\iu} \left\{\sum_{i=0}^{\mathrm{m_{uu}^\nu}-1} \frac{\mathrm{s_u}^i}{i!}e^{-\mathrm{N_0}\mathrm{s_u}} \, \sum_{j=0}^{i} \binom{i}{j} \mathrm{N_0}^{i-j} {\iu}^j e^{-\mathrm{s_u} \iu} \right\} \\ \nonumber
	&= \sum_{i=0}^{\mathrm{m_{uu}^\nu}-1} \q_{\u,i}^\nu \cdot \mathbb{E}_{\iu}\left\{{\iu}^i e^{-\mathrm{s_u} \iu} \right\} \\ 
	&= \sum_{i=0}^{\mathrm{m_{uu}^\nu}-1} (-1)^i \q_{\u,i}^\nu \cdot \D^i_{\mathrm{s_u}} \left[ \lapiu^\nu(\mathrm{s_u}) \right],
	\end{align}
	where (a) is obtained using the CDF of the small-scale fading in \eqref{FadingCDF}. As for the Laplacian of the interference in \eqref{CuRu}, we can write
	\begin{align} \label{lapiu}
	\lapiu^\nu(\mathrm{s_u}) \!=\! \lapicul^\nu(\mathrm{s_u}) \!\cdot\! \lapicun^\nu(\mathrm{s_u}) \!\cdot\! \lapiuul^\nu(\mathrm{s_u}) \!\cdot\! \lapiuun^\nu(\mathrm{s_u}),
	\end{align}
	where $I_\mathrm{xy}^\xi$ is the interference imposed by nodes x of condition $\xi$ on y. Each term in \eqref{lapiu} can be characterized as follows:
	\begin{align}
	\lapixyxi^\nu &= e^{ -2 \pi \lambda_\mathrm{x} \mathcal{I}_\mathrm{xy}^\xi };~~\xi \in \{\L,\N\},
	\end{align}
	where $\lambda_\c=\lamb$ accounts for the density of active GUEs, and
	\begin{align} 
	&\mathcal{I}_\mathrm{xy}^\xi = \int_0^\infty \prxy^\xi(r) \Big(1-\mathbb{E}_{\px,\sixy^\xi} \left[ e^{-\sy \px \,\zetxy^\xi(r)^{-1} \, \sixy^\xi} \right]\Big)\, r \,\mathrm{d}r  \nonumber \\ \label{IxyxiCal}
	&= \sum_{i = 1}^\infty \prxy^\xi(\r_i)\,\, \mathbb{E}_{\px,\sixy^\xi} \left[ \int_{\r_{i}}^{\r_{i+1}}  \Big(1- e^{-\s \px \,d_\mathrm{xy}^{-\axy^\xi} \, \sixy^\xi} \Big)\, r \,\mathrm{d}r \right],
	\end{align}
	where $	\s = \mathrm{s_y} \frac{\mathrm{g_{xy}}(\r_i)}{\hat{\tau}_{\mathrm{xy}}^\xi} $.
	In the following, we calculate the integral term in the right-hand side of \eqref{IxyxiCal}.
	Let us consider a change of variable as $\omega = \s \px \,d_\mathrm{xy}^{-\axy^\xi} \, \sixy^\xi$, which yields
	\begin{equation} \label{eq2}
	\begin{aligned}
	\int_{\r_{i}}^{\r_{i+1}}  &\Big(1- e^{-\s \px \,d_\mathrm{xy}^{-\axy^\xi} \, \sixy^\xi} \Big)\, r \,\mathrm{d}r \\
	&= \frac{(\s \px \sixy^\xi)^{\bxy^\xi}}{\axy^\xi} \int_{ \omega_2}^{ \omega_1} \omega^{-1-\bxy^\xi} (1-e^{-\omega}) \d\omega,
	\end{aligned}
	\end{equation}
	where $\bxy^\xi \triangleq 2/\axy^\xi$, $\omega_1 =  \mu_1 \sixy^\xi$, $\omega_2 =  \mu_2 \sixy^\xi$ and
	\begin{align}
	\mu_1 &\triangleq \frac{\s \px}{(\r_i^2+\mathrm{h_{xy}^2})^{\axy^\xi/2}},~~~	\mu_2 \triangleq \frac{\s \px}{(\r_{i+1}^2+\mathrm{h_{xy}^2})^{\axy^\xi/2}}.
	\end{align}
	The integral in the right-hand side of \eqref{eq2} is equal to
	\begin{equation} \label{eq1}
	\begin{aligned}
	\int_{ \omega_2}^{ \omega_1} &\omega^{-1-\bxy^\xi} (1-e^{-\omega}) \d\omega = \frac{\axy^\xi}{2} \Big[ \omega_2^{-\bxy^\xi} (1-e^{-\omega_2}) \\
	&- \omega_1^{-\bxy^\xi} (1-e^{-\omega_1})  + \int_{\omega_2 }^{\omega_1} \omega^{-\bxy^\xi} e^{-\omega}  \,\mathrm{d}\omega \Big],
	\end{aligned}
	\end{equation}
	in which integration by parts is applied. Also, the integral in the right-hand side of \eqref{eq1} can be written as 
	\begin{align} \label{eq3}
	\int_{\omega_2 }^{\omega_1} \omega^{-\bxy^\xi} e^{-\omega}  \,\mathrm{d}\omega = \gamma\left(1-\bxy^\xi,\omega_1\right) - \gamma\left(1-\bxy^\xi,\omega_2\right),
	\end{align}
	where we used the definition of the incomplete gamma function.
	Therefore, by substituting \eqref{eq3} into \eqref{eq1}, and the corresponding result into \eqref{eq2}, we obtain
	\begin{align}\label{integralStep}
	&\int_{\r_{i}}^{\r_{i+1}}  \Big(1- e^{-\s \px \,d_\mathrm{xy}^{-\axy^\xi} \, \sixy^\xi} \Big)\, r \,\mathrm{d}r \\ \nonumber
	& = \frac{\r_{i+1}^2+\mathrm{h_{xy}^2}}{2}(1-e^{-\mu_2\sixy^\xi}) - \frac{\r_{i}^2+\mathrm{h_{xy}^2}}{2}(1-e^{-\mu_1 \sixy^\xi}) \\ \nonumber
	&+ \frac{(\s \px \sixy^\xi)^{\bxy^\xi}}{2} \Big[  \gamma\left(1-\bxy^\xi,\mu_2\sixy^\xi \right) -  \gamma\left(1-\bxy^\xi,\mu_1 \sixy^\xi \right) \Big].
	\end{align}
	In order to obtain the expectation in the right-hand side of \eqref{IxyxiCal}, we note that for Nakagami-m fading $\psi$ with parameter $\mathrm{m}$ we have
	\begin{align}
	\mathbb{E}_{\psi} \left[e^{-\mu \psi}\right] = \left(1+\frac{\mu}{\mathrm{m}}\right)^{-\mathrm{m}}.
	\end{align}
	Also, by using \cite[eq. 6.455]{gradshteyn2014table} we obtain
	\begin{align} 
	&\mathbb{E}_{\psi} \left[\psi^{\beta} \gamma \left(1-\beta,\mu \psi \right)\right] \nonumber\\
	&= \frac{\mathrm{m}^\mathrm{m}}{\Gamma(\mathrm{m})}  \int_0^\infty \omega^{\beta+\mathrm{m}-1} e^{-\mathrm{m} \omega}  \,\gamma\left(1-\beta,\mu \omega \right) \,\mathrm{d}\omega \nonumber\\ 
		&= \frac{\mathrm{m}^\mathrm{m}}{\Gamma(\mathrm{m})} \cdot \frac{\mu^{1-\beta}\Gamma(\mathrm{m}+1)}{(1-\beta)(\mathrm{m}+\mu)^{1+\mathrm{m}}}\, {_2}F_1\left(1,1+\mathrm{m};2-\beta;\frac{\mu}{\mu+\mathrm{m}}\right) \nonumber \\
	&= \frac{\mathrm{m}^{1+\mathrm{m}} \mu^{1-\beta}}{(1-\beta)(\mathrm{m}+\mu)^{1+\mathrm{m}}}\,{_2}F_1\left(1,1+\mathrm{m};2-\beta;\frac{\mu}{\mu+\mathrm{m}}\right).
	\end{align}
	Now following the transformation properties of the hypergeometric function \cite[eq. 9.131]{gradshteyn2014table} we can write
	\begin{equation}\label{hyperFunc}
	\begin{aligned} 
	&{_2}F_1\left(1,1+\m;2-\beta;\frac{\mu}{\mu+\m}\right) \\
	&= \left(\frac{\m}{\mu+\m}\right)^{-1-\m} {_2}F_1\left(1+\m,1-\beta;2-\beta;-\frac{\mu}{\m}\right).
	\end{aligned}
	\end{equation}
	Therefore by using \eqref{integralStep}--\eqref{hyperFunc} we have
	\begin{equation} \label{eq4}
	\begin{aligned}
	\mathbb{E}_{\sixy^\xi} &\left[ \int_{\r_{i}}^{\r_{i+1}}  \Big(1- e^{-\s \px \,d_\mathrm{xy}^{-\axy^\xi} \, \sixy^\xi} \Big)\, r \,\mathrm{d}r \right] \\ &= \Psi_\mathrm{xy}^\xi\left(\s,\r_{i+1}\right) - \Psi_\mathrm{xy}^\xi\left(\s,\r_{i}\right),
	\end{aligned}
	\end{equation}
	and accordingly by replacing \eqref{eq4} into \eqref{IxyxiCal} we conclude
	\begin{align} \nonumber 
	&\mathcal{I}_\mathrm{xy}^\xi = \mathbb{E}_{\px} \left[ \sum_{i = 1}^\infty \prxy^\xi(\r_i) \Big(\Psi_\mathrm{xy}^\xi\left(\s,\r_{i+1}\right) - \Psi_\mathrm{xy}^\xi\left(\s,\r_{i}\right)\Big) \right] \\ \nonumber
	&= \int_0^\infty  f_{R_\mathrm{x}}^\mathrm{L}(x)\,\sum_{i = 1}^\infty \prxy^{\xi} \Big(\underbrace{\Psi_\mathrm{xy}^\xi\left(\mathrm{s},\r_{i+1}\right) - \Psi_\mathrm{xy}^\xi\left(\mathrm{s},\r_{i}\right)}_\text{computed at $P_\mathrm{x}^\mathrm{L}$}\Big)  \,\mathrm{d}x
	\\ \label{IxyFinal}
	&+ \int_0^\infty  f_{R_\mathrm{x}}^\mathrm{N}(x)\,\sum_{i = 1}^\infty \prxy^{\xi} \Big(\underbrace{\Psi_\mathrm{xy}^\xi\left(\mathrm{s},\r_{i+1}\right) - \Psi_\mathrm{xy}^\xi\left(\mathrm{s},\r_{i}\right)}_\text{computed at $P_\mathrm{x}^\mathrm{N}$}\Big)  \,\mathrm{d}x.
	\end{align}
	Using
	\begin{equation}
	\begin{aligned}
	\sum_{i = 1}^\infty \prxy^\xi(\r_i) &\Big(\Psi_\mathrm{xy}^\xi\left(\s,\r_{i+1}\right) - \Psi_\mathrm{xy}^\xi\left(\s,\r_{i}\right)\Big) \\
	&= \sum_{i = 1}^\infty \left[\prxy^{\xi}(\r_{i-1})-\prxy^{\xi}(\r_{i})\right] \Psi_\mathrm{xy}^\xi\left(\mathrm{s},\r_{i}\right)
	\end{aligned}
	\end{equation}
	in \eqref{IxyFinal} completes the proof.

			\subsection{Calculating the Derivatives of the Laplacian $\lapiu^\nu(\mathrm{s_u})$} \label{Laplacian_derivatives}
		In the following we explain the recursive computation of the Laplacian's dervative.
		According to the formula of Leibniz \cite{roman1980formula}, for the i-th derivative of $\lapiu(\s_\u)$ in (\ref{eq:U2ULinkCondCoverage}) we can write
	\begin{align}
	\D^i_{\s_\u}[\lapiu^\nu(\mathrm{s_u})]  = \sum_{j=0}^{i-1} \binom{i-1}{j} \D^{(i-j)}_{\s_\u}[\Lambda(\mathrm{s_u})] \cdot \D^j_{\s_\u}[\lapiu^\nu(\mathrm{s_u})],
	\end{align}
	where the i-th derivative of $\Lambda(\s_\u)$
	can be written as
	\begin{align} 
	\D^i_{\s_\u}[\Lambda] =  -2 \pi \Big(\lamu \!\!\!\! \sum_{\xi \in \{\mathrm{L},\mathrm{N}\}} \D^i_{\s_\u}[\mathcal{I}_\mathrm{uu}^\xi] + \lamb  \!\!\!\! \sum_{\xi \in \{\mathrm{L},\mathrm{N}\}}\D^i_{\s_\u}[\mathcal{I}_\mathrm{cu}^\xi]\Big),
	\end{align}
	with the i-th derivative of $\mathcal{I}_\mathrm{xy}^\mathrm{\xi}$ being
	\begin{equation} \label{eeq4}
	\begin{aligned} 
	&\D^i_{\s_\u}[\mathcal{I}_\mathrm{xy}^\mathrm{\xi}] \\ 
	&= \int_0^\infty  f_{R_\mathrm{x}}^\mathrm{L}(x)\,\sum_{i = 1}^\infty \left[\prxy^{\xi}(\r_{i-1})-\prxy^{\xi}(\r_{i})\right] \D^i_{\s_\u}[\Psi_\mathrm{xy}^\xi]  \,\mathrm{d}x
	\\ 
	&+ \int_0^\infty  f_{R_\mathrm{x}}^\mathrm{N}(x)\,\sum_{i = 1}^\infty \left[\prxy^{\xi}(\r_{i-1})-\prxy^{\xi}(\r_{i})\right] \D^i_{\s_\u}[\Psi_\mathrm{xy}^\xi]  \,\mathrm{d}x,
	\end{aligned}
	\end{equation}
	and the i-th derivative of $\Psi_\mathrm{xy}^\xi$ being
	\begin{equation} 
	\begin{aligned}
	\D^i_{\s_\u}[\Psi_\mathrm{xy}^\xi] &= -\frac{\r^2+\mathrm{h_{xy}^2}}{2} \D^i_{\s_\u}\left[\left(\frac{\m}{\m+\mu}\right)^{\m}\right] \\
	&\!\!\!\!\!- \D^i_{\s_\u} \left[\mathcal{K}
	\,{_2}F_1\left(1+\m,1-\beta;2-\beta;-\frac{\mu}{\m}\right) \right].
	\end{aligned} \label{derivative_Psi}
	\end{equation}


From \eqref{eeq5} one can see that $\mu$ is a linear function of $\s$ and hence a linear function of $\s_\u$, and therefore can be written as $\mu = \ell_1 \cdot \s_\u$ where $\ell_1$ is a new parameter independent from $\s_\u$. Therefore one can see 
\begin{equation} \label{eeq3}
\begin{aligned}
\D^i_{\s_\u}& \left[\left(\frac{\m}{\m+\mu}\right)^{\m} \right]= \D^i_{\s_\u} \left[\left(1+\ell_2 \cdot \s_\u\right)^{-\m}\right] \\
&= (-1)^i (\m)_i \ell_2^i \left(1+\ell_2 \cdot \s_\u\right)^{-\m-i},
\end{aligned}
\end{equation}
where $(\m)_i \triangleq \frac{(\m+i-1)!}{(\m-1)!}$, and $\ell_2 = \ell_1/m$. 


Also from \eqref{eeq5} we find out that $\mathcal{K}$ has linear dependency on $\s$ and equivalently $\s_\u$, and therefore can be stated as $\mathcal{K} = \ell_3 \cdot \s_\u$ where $\ell_3$ is a parameter with no dependency on $\s_\u$. Thus, one can write 
\begin{equation} \label{eeq2}
\begin{aligned}
\D^i_{\s_\u} &\left[\mathcal{K}
\,{_2}F_1\left(1+\m,1-\beta;2-\beta;-\frac{\mu}{\m}\right) \right] \\
&=\mathcal{K} \D^i_{\s_\u} \left[
\,{_2}F_1\left(1+\m,1-\beta;2-\beta;-\frac{\mu}{\m}\right) \right] \\
&+i \ell_3 \D^{i-1}_{\s_\u} \left[
\,{_2}F_1\left(1+\m,1-\beta;2-\beta;-\frac{\mu}{\m}\right) \right],
\end{aligned}
\end{equation}
where from \cite[eq. 1.29.1]{brychkov2008handbook} we have
\begin{equation} \label{eeq1}
	\begin{aligned}
	\D^i_{\s_\u} &\left[
	\,{_2}F_1\left(1+\m,1-\beta;2-\beta;-\frac{\mu}{\m}\right) \right] \\ &= \left(\frac{-\ell_1}{\m}\right)^i \frac{(\m+1)_i (1-\beta)_i}{(2-\beta)_i} \\ &\times {_2}F_1\left(1+\m+i,1-\beta+i;2-\beta+i;-\frac{\mu}{\m}\right).
	\end{aligned}
\end{equation}


By using \eqref{eeq1}, we obtain \eqref{eeq2}. Subsequently, \eqref{eeq4} can be computed by substituting \eqref{eeq2} and \eqref{eeq3} into \eqref{derivative_Psi}, which completes the recursive computation of the Laplacian's derivative. 

		
		
	\subsection{Proof of Theorem \ref{C2Bcov_theorem}} \label{C2Bcov_proof}
	
	To obtain the GUE UL coverage we can write
	\begin{align} \nonumber
	\pcovc &= \mathbb{P}\left[\frac{\pc \zcb^{-1} \, \scb}{\mathrm{N_0} + I_\c} > \t \right]
	\\ \label{GUEcovExpr}
	& = \sum_{\nu \in \{\mathrm{L},\mathrm{N}\}}\int_0^\infty  \mathcal{C}_{\c|\Rc}^\nu(\rc)\,f_{\Rc}^\nu(\rc) \,\mathrm{d}\rc,
	\end{align}
	where, similarly to \eqref{CuRu}, we have
	\begin{align} \nonumber
	\mathcal{C}_{\c|\Rc}^\nu(\rc) &\triangleq \mathbb{P}\left[\frac{\pc^\nu \zcb^\nu(\rc)^{-1} \, \scb^\nu}{\mathrm{N_0} + I_\c} > \t \right] \\ \label{condGUEcoveExpr}
	&= \sum_{i=0}^{\mathrm{m_{gb}^\nu}-1} (-1)^i \q_{\c,i}^\nu \cdot \D_{\s_\c}^i\left[ \lapic^\nu(\s_\c) \right].
	\end{align}
		
	The Laplacian of the aggregate interference, i.e. $\lapic^\nu(\s_\c)$ in \eqref{condGUEcoveExpr}, can be derived as follows
	\begin{align} \label{eq11}
	\lapic^\nu(\s_\c) = \lapiucl^\nu(\s_\c) \cdot \lapiucn^\nu(\s_\c) \cdot \lapiccl^\nu(\s_\c) \cdot \lapiccn^\nu(\s_\c),
	\end{align}
	where $\lapiucl$ and $\lapiucn$ are obtained similarly to \eqref{IxyFinal}. To characterize the interference from other GUEs, i.e. $I_{\mathrm{gg}}^{\xi}$, we can write
	\begin{align} \label{eq10}
	\mathcal{L}_{I_\mathrm{gg}^\xi} &= e^ {-2 \pi \int_0^\infty \lamci(r) \Big(1-\mathbb{E}_{\pc,\scb^\xi} \left[ e^{-\s_\c \pc \,\zcb^\xi(r)^{-1} \, \scb^\xi} \right]\Big)\, r \,\mathrm{d}r },
	\end{align}
	which can be stated as $\mathcal{L}_{I_\mathrm{gg}^\xi} = e^{-(2\pi\lamb)^2 \,\mathcal{I}_\mathrm{gg}^\xi}$ with
	\begin{equation} \label{eq9}
	\begin{aligned}
	&\mathcal{I}_\mathrm{gg}^\xi = \sum_{\nu \in \{\L,\N\}} \int_0^\infty \prcb^\xi(r) \times\\
	& \int_0^r \prcb^\nu(x) x e^{-\lamb \pi x^2}  \Bigg(1-\mathbb{E}_{\scb^\xi} \Bigg[ e^{-\frac{\s_\c \pc^\nu(x)  \, \scb^\xi}{\zcb^\xi(r)}} \Bigg]\Bigg)\,\mathrm{d}x\,  r \,\mathrm{d}r.
	\end{aligned}
	\end{equation}
	We rewrite the above integral as
	\begin{equation} \label{eq8}
	\begin{aligned}
	\int_0^\infty \prcb^\xi(r) &\int_0^r \prcb^\nu(x) x e^{-\lamb \pi x^2} \\ &\times \Bigg(1-\mathbb{E}_{\scb^\xi} \Bigg[ e^{-\frac{\s_\c \pc^\nu(x)  \, \scb^\xi}{\zcb^\xi(r)}} \Bigg]\Bigg)\,\mathrm{d}x\,  r \,\mathrm{d}r \\
	&\hspace{-2cm}= \int_0^\infty \prcb^\nu(x) x e^{-\lamb \pi x^2} \int_x^\infty  \prcb^\xi(r) \\ &\times \Bigg(1-\mathbb{E}_{\scb^\xi} \Bigg[ e^{-\frac{\s_\c \pc^\nu(x) \, \scb^\xi}{\zcb^\xi(r)}} \Bigg]\Bigg)\,r \,\mathrm{d}r\mathrm{d}x,
	\end{aligned}
	\end{equation}
	where the inner integral can be derived as follows
	\begin{align} \nonumber
	&\int_x^\infty  \prcb^\xi(r) \Bigg(1-\mathbb{E}_{\scb^\xi} \Bigg[ e^{-\frac{\s_\c \pc^\nu(x) \, \scb^\xi}{\zcb^\xi(r)}} \Bigg]\Bigg)\,r \,\mathrm{d}r \\ \nonumber
	&= \sum_{i = j(x)}^\mathrm{\infty} \prcb^\xi(\r_i)\, \mathbb{E}_{\scb^\xi} \Bigg[ \int_{\r_{i}}^{\r_{i+1}}  \Bigg(1- e^{-\s \pc^\nu \,\dcb^{-\acb^\xi} \, \scb^\xi} \Bigg)\, r \,\mathrm{d}r \Bigg] \\ \label{eq7}
	&= \sum_{i = j(x)}^\mathrm{\infty} \prcb^\xi(\r_i)\, \Bigg(\underbrace{\Psi_\mathrm{gb}^\xi\left(\s,\r_{i+1}\right) - \Psi_\mathrm{gb}^\xi\left(\s,\r_{i}\right)}_{\text{at $\pc = \pc^\nu(x)$}} \Bigg)
	\end{align}
	with $\mathrm{s} = \s_\c \frac{g_\mathrm{{gb}(\r_i)}}{\hat{\tau}_\mathrm{gb}^\xi}$. Note that we have approximated the BS antenna gain as invariant within $[r_{i},r_{i+1}]$, so that $g_\mathrm{gb}(r) = g_\mathrm{gb}(\r_i)$ is a constant value. Such approximation holds tight as the interval can be chosen as arbitrarily small.	
	
Finally, \eqref{eq9} can be calculated by substituting \eqref{eq7} into \eqref{eq8}, and it can then be used in \eqref{eq10} to compute the Laplacian of the interference in \eqref{eq11}. Subsequently, using \eqref{eq11} in \eqref{condGUEcoveExpr}, and the corresponding result in \eqref{GUEcovExpr} concludes the proof.

\end{appendix}

\subsection{Proof of Corollary~\ref{corollary:U2Ucoverage}} \label{proof:U2Ucoverage_approx}

From Approximation~2, we have $\mathcal{C}_{\mathrm{u}|\Ru}^\N(\ru) = 0$, thus
\begin{align}  \label{eq:U2ULinkCoverageApp}
\pcovuav &= \sum_{\nu \in \{\mathrm{L},\mathrm{N}\}}\int_0^{\rMuu}  \mathcal{C}_{\mathrm{u}|\Ru}^\nu(\ru)\,f_{\Ru}^\nu(\ru) \,\mathrm{d}\ru \nonumber\\
&= \int_0^\rMuu f_{\Ru}^\L(\ru) \mathcal{C}_{\mathrm{u}|\Ru}^\L(\ru) \mathrm{d}\ru,
\end{align}
where by using Approximation~1 we can write
\begin{align} \label{CuRuApp} 
&\mathcal{C}_{\mathrm{u}|\Ru}^\L(\ru) 
= \mathbb{E}_{I_\mathrm{u}} \left\{\mathbb{P}\left[\suu^\L > \frac{\t}{\pu^\L \zuu^\L(\ru)^{-1}}(\mathrm{N_0} + I_\mathrm{u}) \right] \right\} \nonumber\\ 
&= 1 -  \mathbb{E}_{I_\mathrm{u}} \left\{\mathbb{P}\left[\suu^\L < \frac{\t}{\pu^\L \zuu^\L(\ru)^{-1}}(\mathrm{N_0} + I_\mathrm{u}) \right] \right\} \nonumber\\ 
&\approx \mathbb{E}_{\iu} \left\{\sum_{i=1}^{\mathrm{m_{uu}^\L}} \binom{\mathrm{m_{uu}^\L}}{i}(-1)^{i+1} e^{-z_{\u,i}^\L (\mathrm{N_0} + I_\mathrm{u})} \right\} \nonumber\\ 
&= \sum_{i=1}^{\mathrm{m_{uu}^\L}} \binom{\mathrm{m_{uu}^\L}}{i}(-1)^{i+1} e^{-z_{\u,i}^\L \mathrm{N_0}} \cdot \mathbb{E}_{\iu}\left\{e^{-z_{\u,i}^\L I_\mathrm{u}} \right\} \nonumber\\ 
&= \sum_{i=1}^{\mathrm{m_{uu}^\L}} \binom{\mathrm{m_{uu}^\L}}{i}(-1)^{i+1} e^{-z_{\u,i}^\L \mathrm{N_0}} \cdot \lapiu^\L(z_{\u,i}^\L).
\end{align}

Under Approximation~2, we can neglect the interference generated by NLoS links and obtain
\begin{equation} \label{eqn:LaplacianProofTheoremOne}
\lapiu^\L(z_{\u,i}^\L) = e^{ -2 \pi (\hat{\lambda}_\u \mathcal{I}_\mathrm{uu}^\L + \lamb \mathcal{I}_\mathrm{gu}^\L)}.
\end{equation}

Corollary~\ref{corollary:U2Ucoverage} then follows by deriving $\mathcal{I}_\mathrm{uu}^\L$ and $\mathcal{I}_\mathrm{gu}^\L$ from \eqref{IxyFinal} by replacing $\pu$ with its mean, and by substituting $\mathcal{I}_\mathrm{uu}^\L$ and $\mathcal{I}_\mathrm{gu}^\L$ into (\ref{eqn:LaplacianProofTheoremOne}), (\ref{CuRuApp}), and (\ref{eq:U2ULinkCoverageApp}).


\subsection{Proof of Proposition~\ref{proposition:meanUAVtxPower}} \label{proof:meanUAVtxPower}

The mean UAV transmit power can be written as
\begin{align} \label{eqn:meanPowerProof1}
\mathbb{E}[\pu] = \sum_{\nu \in \{\mathrm{L},\mathrm{N}\}}\int_0^\rMuu f_{\Ru}^\nu(\ru) \mathbb{E}\left[\pu^\nu|\Ru = \ru\right]  \mathrm{d}\ru,
\end{align}
where $f_{\Ru}^\nu(\ru) = f_{\Ru}(\ru) \cdot \pruu^\nu(\ru)$ and where the integral in (\ref{eqn:meanPowerProof1}) can be written as
\begin{equation} \label{eqn:meanPowerIntegrals}
\begin{aligned}
&\int_0^\rMuu f_{\Ru}^\nu(\ru) \mathbb{E}\left[\pu^\nu|\Ru = \ru\right]  \mathrm{d}\ru \\
&\!=\!\int_0^\mathrm{r_m^\nu} \! \rho_\u \zeta_{\u\u}^{\eu} \cdot f_{\Ru}^\nu(\ru) \, \mathrm{d}\ru \!+\! \int_\mathrm{r_m^\nu}^\rMuu \! \pumax \cdot f_{\Ru}^\nu(\ru) \mathrm{d}\ru. \! 
\end{aligned}
\end{equation}
The first integral on the right-hand side of (\ref{eqn:meanPowerIntegrals}) is equal to
\begin{equation}
\begin{aligned}
&\int_0^\mathrm{r_m^\nu} \! \rho_\u \zeta_{\u\u}^{\eu} \cdot f_{\Ru}^\nu(\ru) \, \mathrm{d}\ru \\
&=\sum_{i=1}^{j} c_i^\nu \int_{\r_i}^\mathrm{r_{i+1}} \ru^{1+\auu^\nu \eu } \cdot e^{-\r_\u^2/(2\sigma_\mathrm{u}^2)}  \mathrm{d}\ru
\end{aligned}
\end{equation}
where 
\begin{align}
c_i^\nu &= \frac{\pur\left(\hat{\tau}_{\mathrm{uu}}^\nu/\guu\right)^{\eu}}{\sigma_\u^2[1-\e^{-r_\mathrm{M}^2/(2\sigma_\mathrm{u}^2)}]} \cdot \pruu^\nu(r_i).
\end{align}
With the change of variable $y = \r_\u^2/2\sigma_\mathrm{u}^2$, we can write
\begin{align} \label{eqn:meanPowerFirstIntegral}
&c_i^\nu \int_{\r_i}^\mathrm{r_{i+1}} \ru^{1+\auu^\nu \eu } \cdot e^{-\r_\u^2/(2\sigma_\mathrm{u}^2)}  \mathrm{d}\ru \nonumber\\
&= C_i^\nu \int_{y_i}^{y_{i+1}} y^{\auu^\nu \eu /2} \cdot e^{-y}  \mathrm{d}y \nonumber\\
&= C_i^\nu \left(\int_{0}^{y_{i+1}} y^{\auu^\nu \eu /2} \cdot e^{-y}  \mathrm{d}y - \int_{0}^{y_{i}} y^{\auu^\nu \eu /2} \cdot e^{-y}  \mathrm{d}y \right) \nonumber\\
& \!= C_i^\nu \Big[ \gamma(1\!+\!\auu^\nu \eu /2,y_{i+1}) \!-\! \gamma(1\!+\!\auu^\nu \eu /2,y_{i}) \Big] 
\end{align}
where $y_i = \frac{r_i^2}{2\sigma_u^2}$ and
\begin{align}
C_i^\nu &= \frac{(2\sigma_\u^2)^{{\auu^\nu\eu/2}}\pur\left(\hat{\tau}_{\mathrm{uu}}^\nu/\guu\right)^{\eu}}{1-\e^{-r_\mathrm{M}^2/(2\sigma_\mathrm{u}^2)}} \cdot \pruu^\nu(r_i);~i>0,
\end{align}
thus obtaining
\begin{equation} \label{eqn:meanPowerFirstIntegralResult}
\begin{aligned}
&\int_0^\mathrm{r_m^\nu} f_{\Ru}^\nu(\ru) \mathbb{E}[\rho_\u \zeta_{\u\u}^{\eu}] \, \mathrm{d}\ru \\
& = \sum_{i=1}^{j} [C_i^\nu-C_{i+1}^\nu] \, \gamma(1+\auu^\nu \eu k/2,y_{i+1})
\end{aligned},
\end{equation}
where $C_{j+1}^\nu = 0$. 

Similarly, the second integral on the right-hand side of (\ref{eqn:meanPowerIntegrals}) is equal to
\begin{equation} \label{eqn:meanPowerSecondIntegral}
\int_\mathrm{r_m^\nu}^\rMuu  \pumax \cdot f_{\Ru}^\nu(\ru) \mathrm{d}\ru = \sum_{i=j+1}^{k+1} [B_i^\nu-B_{i-1}^\nu] \, e^{-r_i^2/(2\sigma_\mathrm{u}^2)},
\end{equation}
where $B_j^\nu = 0$, $B_{k+1}^\nu = 0$, and
\begin{align}
B_i^\nu &= \frac{\pumax \, \pruu^\nu(r_i) }{1-\e^{-r_\mathrm{M}^2/(2\sigma_\mathrm{u}^2)}};~i>j.
\end{align}
Proposition~\ref{proposition:meanUAVtxPower} then follows from substituting (\ref{eqn:meanPowerFirstIntegralResult}) and (\ref{eqn:meanPowerSecondIntegral}) into (\ref{eqn:meanPowerIntegrals}), and then into (\ref{eqn:meanPowerProof1}).

\balance
\ifCLASSOPTIONcaptionsoff
  \newpage
\fi
\bibliographystyle{IEEEtran}
\bibliography{Strings_Gio,Bib_Gio,Bib_Mahdi}
\end{document}